\documentclass[11pt]{article}

\usepackage[letterpaper,margin=1in]{geometry}
\usepackage{amsmath, amssymb, amsthm, amsfonts}
\usepackage{bbm}
\usepackage{bm}
\usepackage{complexity}

\usepackage{subcaption}

\usepackage{ifthen}
\usepackage{comment}
\usepackage{xfrac}
\usepackage{tikz}
\usetikzlibrary{positioning,decorations.pathreplacing}
\usepackage{comment}

\usepackage{cite}
\usepackage{appendix}
\usepackage{graphicx}
\usepackage{color}
\usepackage{algorithm}
\usepackage[noend]{algpseudocode}
\usepackage{epstopdf}
\usepackage{wrapfig}
\usepackage{paralist}
\usepackage[textsize=tiny]{todonotes}

\usepackage{framed}
\usepackage[framemethod=tikz]{mdframed}
\usepackage[bottom]{footmisc}
\usepackage{enumitem}
\setitemize{noitemsep,topsep=3pt,parsep=3pt,partopsep=3pt}
\usepackage[font=small]{caption}
\usepackage{xspace}
\usepackage[strict]{changepage}
\usepackage{rotating}

\definecolor{darkgreen}{rgb}{0,0.5,0}
\usepackage{hyperref}
\hypersetup{
    unicode=false,          
    colorlinks=true,        
    linkcolor=red,          
    citecolor=darkgreen,        
    filecolor=magenta,      
    urlcolor=cyan           
}
\usepackage[capitalize]{cleveref}

\newtheorem{theorem}{Theorem}[section]
\newtheorem{lemma}[theorem]{Lemma}
\newtheorem{meta-theorem}[theorem]{Meta-Theorem}

\newcommand{\eps}{\varepsilon}

\newcommand{\exclude}[1]{}

\newcommand{\FullOrShort}{full}

\ifthenelse{\equal{\FullOrShort}{full}}{

  \newcommand{\fullOnly}[1]{#1}	
  \newcommand{\shortOnly}[1]{}
  }{
    \newcommand{\fullOnly}[1]{}
	\newcommand{\shortOnly}[1]{#1}
  }
  
\newcommand{\ITWonly}[1]{}
\newcommand{\noITW}[1]{#1}

\begin{document}

\date{}

\title{Rate-Distance Trade-offs for List-Decodable Insertion-Deletion Codes}



\author{Bernhard Haeupler\footnote{Supported in part by NSF grants CCF-1527110, CCF-1618280, CCF-1814603, CCF-1910588, NSF CAREER award CCF-1750808, a Sloan Research Fellowship, and funding from the European Research Council (ERC) under the European Union's Horizon 2020 research and innovation program (ERC grant agreement 949272).} \\Carnegie Mellon University \& ETH Zurich\\ \texttt{haeupler@cs.cmu.edu} \and Amirbehshad Shahrasbi\footnotemark[1] \footnote{Supported in part by CRA Computing Innovation Postdoctoral Fellowship.}\\Microsoft\\ \texttt{ashahrasbi@microsoft.com}}

\maketitle
\thispagestyle{empty}

\begin{abstract}
This paper presents general bounds on the highest achievable rate for list-decodable insertion-deletion codes. In particular, we give novel outer and inner bounds for the highest achievable communication rate of any insertion-deletion code that can be list-decoded from any $\gamma$ fraction of insertions and any $\delta$ fraction of deletions. Our bounds simultaneously generalize the known bounds for the previously studied special cases of insertion-only, deletion-only, and zero-rate and correct other bounds that had been reported for the general case. 



\end{abstract}
	
\thispagestyle{empty}

\newpage
\setcounter{page}{1}

\section{Introduction}
Error-correcting codes are classic combinatorial objects that have been extensively studied since late 40s with broad applications in a multitude of communication and storage applications. While error-correcting codes are mostly studied within the setting that concerns symbol substitutions and erasures (i.e., Hamming-type errors), there has been a recent rise of interest in codes that correct from synchronization errors, such as insertions and deletions, from both theoretical~\cite{brakensiek2016efficient,bukh2017improved,guruswami2016efficiently,guruswami2017deletion,haeupler2017synchronization,haeupler2017synchronization3,haeupler2017synchronization2, haeupler2019near,cheng2019synchronization, liu2019list,liu2019explicit,  haeupler2019optimal, cheng2019block, cheng2018deterministic,guruswami2019optimally,cheng2020efficient,cheraghchi2019overview} and practical perspectives~\cite{organick2017scaling,blawat2016forward,goldman2013towards,church2012next,yazdi2015dna,bornholt2016dna}. Such codes and their relevant qualities are defined in the same fashion as error-correcting codes, except that the minimum distance requirement is with respect to the pairwise \emph{edit distance} between code words.

Compared to error-correcting codes for Hamming errors synchronization codes are far less understood and many fundamental questions about them remain to be explored. One such important question is the rate-distance trade-off for (worst-cases) synchronization errors, i.e., determining the largest rate that any synchronization code can achieve in the presence of a certain amount of synchronization errors. We address this question in the list-decoding setting.

A code is list-decodable if there exists a decoder $D$ which, for any corrupted codeword (within the desired error bounds), outputs a small size list of codewords that is guaranteed to include the uncorrupted codeword.
More formally, an insertion-deletion code $C \subseteq \Sigma^n$ (or insdel code, for short) is $(\gamma,\delta,L)$-list-decodable if there exists a function $D:\Sigma^*\to 2^C$ such that $|D(w)| \leq L$ for every $w \in \Sigma^*$ and for every codeword $x \in C$ and every word $w$ obtained from $x$ by at most $\gamma \cdot n$ insertions and at most $\delta\cdot n$ deletions, it is the case that $x \in D(w)$. The parameter $L$ is called the list-size.
These definitions naturally extend to families of codes with increasing block lengths in the usual way: A family of codes is $(\gamma,\delta, L(\cdot))$-list decodable if  each member of the family is $(\gamma,\delta, L(n))$-list decodable where $n$ denotes the block length. Often the function $L$ is omitted and a family of $q$-ary codes $\mathcal{C}$ is said to be $(\gamma,\delta)$-list decodable if there exist some polynomial function $L(\cdot)$ for which $\mathcal{C}$ is $(\gamma,\delta, L(\cdot))$-list decodable. The rate $R$ of a family of $q$-ary codes $\mathcal{C}$ is defined as $R = \lim_{n \rightarrow \infty} \frac{\log_q |C_n|}{n}$.

The fundamental question studied in this paper is to understand the inherent trade-off between the communication rate of a $q$-ary list-decodable insdel code and the amount of synchronization errors it can correct, i.e., the error parameters $\gamma$ and $\delta$. For every fixed alphabet size $q$, this trade-off can be nicely plotted as a 3D-surface in a 3D-chart which plots the maximum communication rate on the $z$-axis for all $\gamma$ and $\delta$ (plotted on the x- and y-axes respectively). \noITW{See \Cref{fig:q=5-intro-simple} for an example of such a 3D-plot.}\ITWonly{See \Cref{fig:q=5-intro} for an example of such a 3D-plot.}

\global\def\FigureIntroSimple{
 \begin{figure}[h]
     \centering
     \includegraphics[scale=.5]{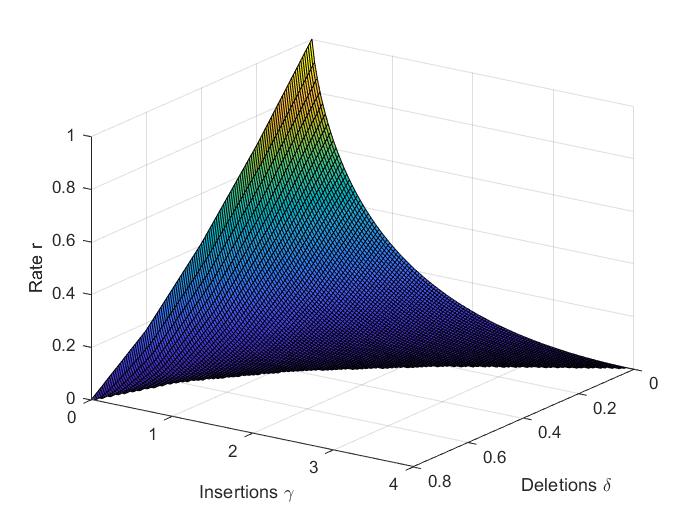}
     \caption{Depiction of our outer bound for $q=5$.}
     \label{fig:q=5-intro-simple}
 \end{figure} 
}

\global\def\FancyPicture{
\noITW{
\begin{sidewaysfigure}
    \centering
    \includegraphics[width=1\textwidth]{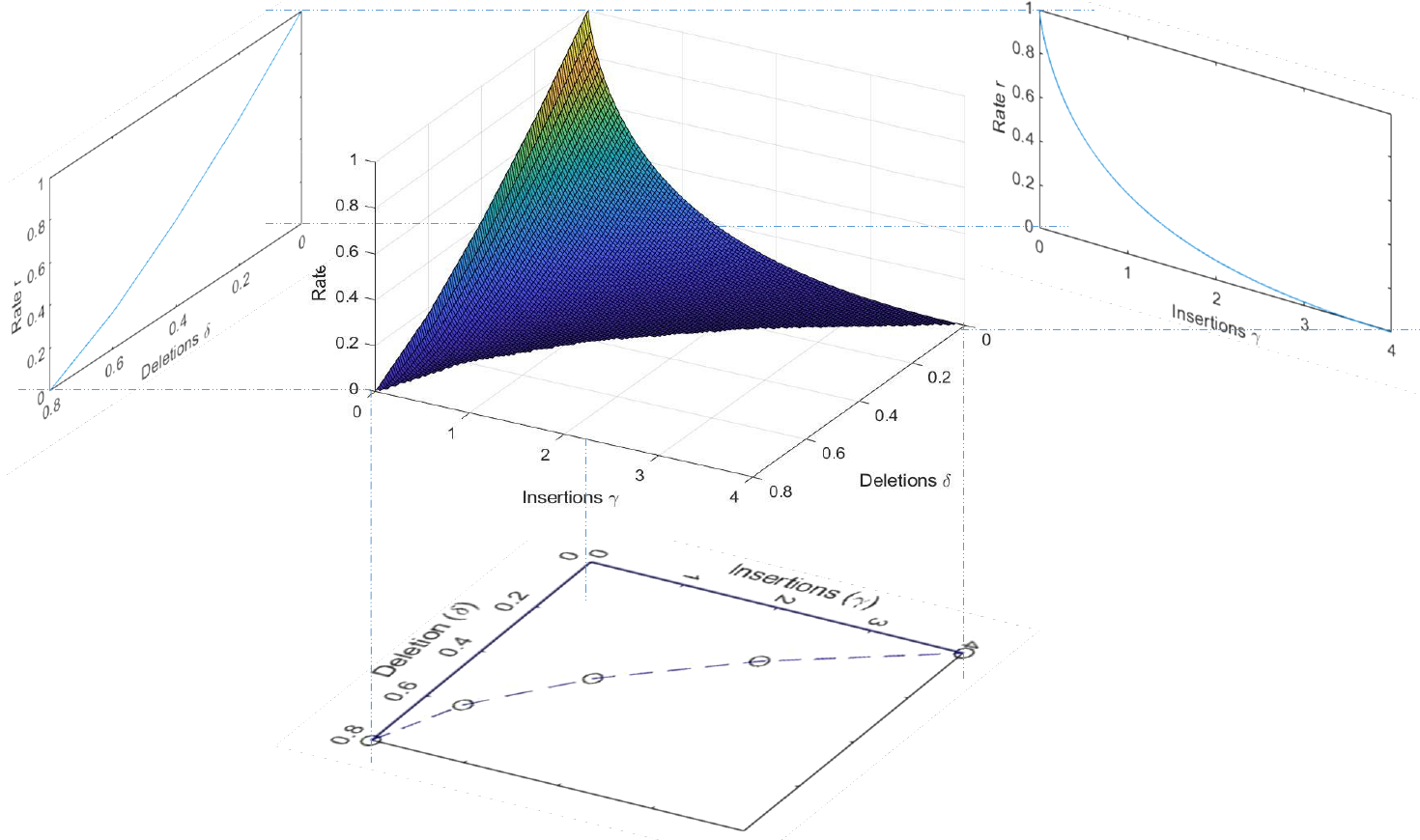}
    \caption{Depiction of our outer bound for $q=5$ and its three projections for the insertion-only case (on the right), the deletion-only case (on the left), and the zero-rate case (on the bottom). The projection graphs are exactly matching the state-of-the-art results of \cite{haeupler2018synchronization4,guruswami2019optimally}.}
    \label{fig:q=5-intro}
\end{sidewaysfigure}
}\ITWonly{\captionsetup{belowskip=-15pt}
\begin{figure}
    \centering
    \includegraphics[width=.95\linewidth]{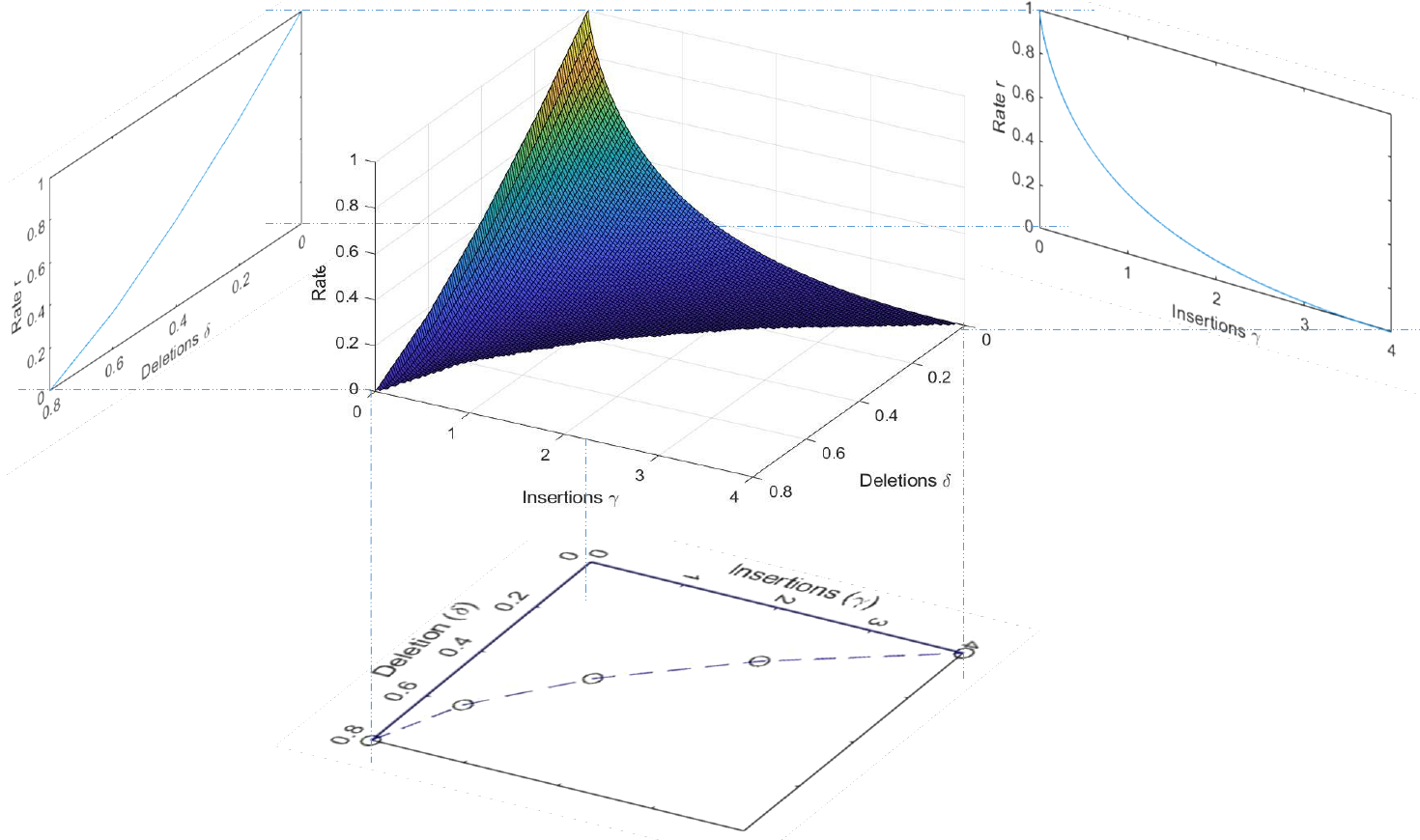}
    \caption{\footnotesize Depiction of our outer bound for $q=5$ and its three projections for the insertion-only case (on the right), the deletion-only case (on the left), and the zero-rate case (on the bottom). The projection graphs are exactly matching the state-of-the-art results of \cite{haeupler2018synchronization4,guruswami2019optimally}.}
    \label{fig:q=5-intro}
\end{figure}

}
}
\noITW{\FigureIntroSimple}
\fullOnly{
\FancyPicture
}

Of course, determining the exact communication rate values for any $q$ and any non-trivial values of $(\gamma,\delta)$ is beyond the capability of current techniques. Prior work (described in \cref{sec:prior-work}) has furthermore mainly focused on obtaining a better understanding of certain special cases, which correspond to projections or cuts of the general trade-off plot. In particular, \Cref{fig:q=5-intro} shows the 2D cuts/projections onto the xz- and yz-planes, which correspond to the insertion-only setting with $\delta=0$ and the deletion-only setting with $\gamma=0$, as well as the projection/cut onto the yz-plane specifying for which error rate combinations of $\gamma$ and $\delta$ the communication rate hits zero. See \Cref{fig:q=5-intro} for examples of these three 2D-projections. It has also been studied how the shape of the 3D-plot changes asymptotically as $q$ gets larger. 

\subsection{Our Results}

This paper is among the first to give results for the entirety of the 3D trade-off between communication rate and the two error rates for every fixed alphabet size $q$. We primarily focus on giving good outer bounds, i.e., impossibility results proving limits on the best possible communication rate (for any given $\gamma,\delta$, and $q$). The novel outer bounds we prove are given in the 3D-plot of \Cref{fig:q=5-intro} for an alphabet size of $q = 5$ (similar plots for any given $q$ apply). We develop these outer bounds in \cref{sec:outer bound}. (See \cref{thm:outer bound})

A notable property of our new outer bound is that, for every $q$, it exactly matches the best previously known results on all three aforementioned projections/cuts. That is, the outer bound implied for deletion-only codes (i.e., the cut on $\gamma = 0$ plane) matches the deletion-only bound from \cite{haeupler2018synchronization4}. Similarly, we match the best insertion-only bounds known (also from \cite{haeupler2018synchronization4}) when restricting or projecting our new general outer bound result to the $\delta = 0$ plane. Finally, the error resilience implied by our bound (i.e., where the curve in \Cref{fig:q=5-intro} hits the floor) precisely matches the list-decoding error resilience curve for insertions and deletions as identified by \cite{guruswami2019optimally}. As such, our bound fully encapsulates and truly generalizes the entirety of the current state-of-the-art of the fundamental rate-distance trade-off for list-decodable insertion-deletion codes\noITW{ for any fixed alphabet size $q$}. 

Lastly, for the sake of completeness and as a comparison point, we also provide a general inner bound in \cref{sec:inner_bound}. This general existence result is obtained by analyzing the list-decodability of random codes. We do this mainly through a simple bound on the size of the insertion-deletion sphere. \noITW{\Cref{fig:q=5-upper-and-lower} illustrates this inner bound in contrast to the outer bound depicted in \Cref{fig:q=5-intro}, also for $q=5$.} It is worth noting that, in contrast to our outer bound, the cut onto the xy-plane does not match the precise error resilience identified (through matching inner and outer bounds) in \cite{guruswami2019optimally}. However, the cuts onto the xz and yz planes do match the inner bounds of \cite{haeupler2018synchronization4} for insertion-only and deletion-only cases, which were also derived by analysis of random codes. We generally believe our outer bounds to be closer to the true zero-error list-decoding channel capacity.


\global\def\FigureInnerAndOuter{
\begin{figure}
\begin{subfigure}{.5\textwidth}
  \centering
  \includegraphics[width=1.06\linewidth]{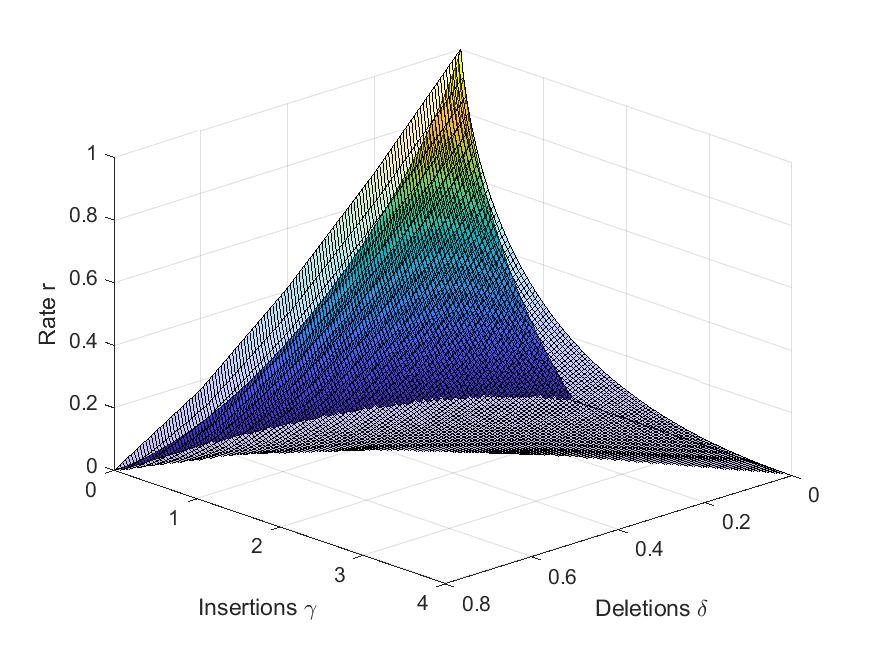}
\end{subfigure}%
\begin{subfigure}{.5\textwidth}
  \centering
  \includegraphics[width=1.06\linewidth]{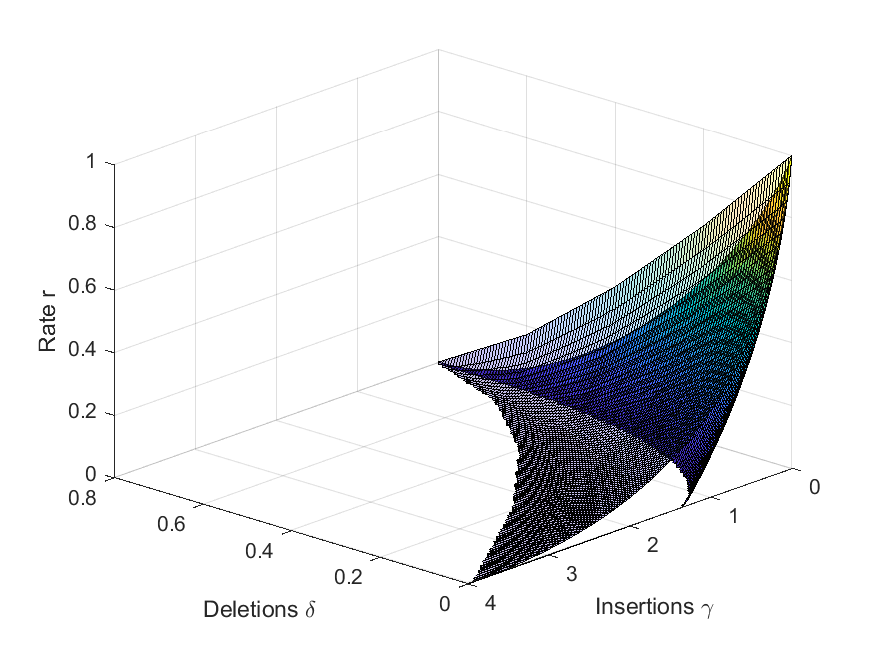}
\end{subfigure}
\caption{Depiction of our inner and outer bounds for $q=5$ from two angles. The more transparent surface is the outer bound of \cref{sec:outer bound} and the surface underneath is the inner bound derived in \cref{sec:inner_bound}.}
\label{fig:q=5-upper-and-lower}
\end{figure}
}
\noITW{\FigureInnerAndOuter}

\subsection{Related Work}\label{sec:prior-work}
This paper studies the fundamental rate-distance trade-off for error correcting codes which are capable of list-decoding from worst-cases insertions and deletions, a topic which has attracted significant attention over the last three years~\cite{guruswami2017deletion,wachter2017list,hayashi2018list,haeupler2018synchronization4,guruswami2019optimally,liu2019list}. We summarize these prior works in detail in this section. The multitude of related work on similar questions, such as, (efficient) list-decoding from Hamming errors, unique-decodable insdel codes, or decoding from random insertions or deletions are too many to list or discuss here. Instead, we refer the interested reader to the following (recent) surveys \cite{cheraghchi2019overview,mitzenmacher2009survey,guruswami2004list,mercier2010survey,haeupler2020survey}, which give detailed accounts of such works. 

As noted above, with the exception of~\cite{liu2019list}, mostly special cases of the general rate-distance trade-off for list-decodable insdel codes have been studied up to now. This includes in particular (combinations of) the deletion-only case (with $\gamma=0$), the insertion-only case (with $\delta=0$), the zero-rate regime or resilience case asking for what extremal values of $(\gamma,\delta)$ a non-zero rate can be obtained, and the case of large alphabets where the alphabet size $q=O(1)$ is allowed to be a large constant that can depend on the error rates $(\gamma, \delta)$.

\subsubsection{List-Decodable Insdel Codes Over Large Constant-Size Alphabets}
The rate-distance tradeoff for list-decodable error correcting codes has been studied in \cite{haeupler2018synchronization4} under the large alphabet setting, that is the question of finding the largest possible achievable rate that $(\gamma,\delta)$-list-decodable families of codes can achieve as long as their alphabet size is constant $q=O_{\gamma, \delta}(1)$ (i.e., independent of the block length). 
Using a method of constructing insdel codes by indexing ordinary error-correcting codes with synchronization strings introduced in \cite{haeupler2017synchronization}, \cite{haeupler2018synchronization4} shows the following: For every $\delta \in (0, 1)$, $\gamma \geq 0$, and sufficiently small $\eps > 0$, there exists an efficient family of $(\gamma,\delta)$-list-decodable codes over an alphabet of size $q=O_{\gamma, \delta, \eps}(1)$ that achieve a rate of $1-\delta-\eps$ or more. It is easy to verify that no such family of codes can achieve a rate larger than $1-\delta$.

The result of \cite{haeupler2018synchronization4} points out an interesting and indeed very drastic distinction between insertions and deletions in the list-decoding setting. In the  unique-decoding setting the effect of insertions and deletions are symmetric, and the rate-distance tradeoff can be fully measured solely in terms of the edit-distance between codewords. For list-decoding it turns out that insertions behave completely different than deletions. Indeed while any $\delta$ fraction of deletions will definitely reduce the rate at the very least to $1 - \delta$, in the very extreme the impact of insertion errors can be fully compensated by taking the alphabet appropriately large. This is what makes the maximum achievable rate for arbitrarily large constant alphabets merely a function of the deletion error rate $\delta$. This stark distinction in the effects insertions and deletions have on the rate-distance tradeoff for list-decodable codes is the reason why it is crucial to use the two parameters $\gamma$ and $\delta$ to keep track of insertions and deletions separately.

\subsubsection{Error Resilience of List-Decodable Insdel Codes}

An important special case of the rate-distance trade-off for list-decodable insertion-deletion codes is the question of the best possible \emph{error resilience}. In particular, the question of \emph{``what is the ``largest" fraction of errors against which list-decoding is possible for some positive-rate code''}\noITW{ -- or differently speaking, at what point(s) the maximum achievable rate of a list-decodable code becomes zero}. Understanding this question is, in some way, a prerequisite to meaningfully talk about more general positive rates. Nevertheless, even when restricted to binary deletions-only or insertions-only codes, finding good bounds on the error resilience is highly non-trivial (in contrast to the Hamming case)~\cite{guruswami2017deletion,guruswami2019optimally,hayashi2018list,guruswami2022zero} and has only recently been solved~\cite{guruswami2019optimally}. (along with the general case where insertions and deletions occur together.)

For deletion-only codes, i.e., the special case where $\gamma = 0$, Guruswami and Wang~\cite{guruswami2017deletion} gave binary codes that are list-decodable from a $\delta = \frac{1}{2} - \epsilon$ fraction of errors attaining a $\poly(\epsilon)$ rate for any $\epsilon > 0$. This implies that the error resilience for deletion coding is precisely $\delta_0 = \frac{1}{2}$ since, with a fraction of deletions $\delta \geq \frac{1}{2}$, an adversary can simply eliminate all instances of the least frequent symbol and convert any codeword from $\{0, 1\}^n$ into either $0^{n/2}$ or $1^{n/2}$.

In 2017 a work of Wachter-Zeh~\cite{wachter2017list} gave Johnson-type bounds on list-decodability and list-sizes of codes given their minimum edit-distance. In 2018, Hayashi and Yasunaga~\cite{hayashi2018list} made corrections to the results presented in \cite{wachter2017list}, and further showed that such bounds give novel results for the insertion-only case of resilience. In particular, they prove that the codes introduced by Bukh, Guruswami, and H\aa stad~\cite{bukh2017improved} can be list-decoded from up to $\gamma = 0.707$ fraction of insertions (and no deletions) while maintaining a positive-rate. 

Very recently, Guruswami et al.~\cite{guruswami2019optimally} improved this fraction of insertions to an optimal $\gamma < 1$. Much more generally \cite{guruswami2019optimally} were able to tightly and fully identify the error resilience region for codes that are list-decodable from a mixture of insertions and deletions, i.e., determine exactly and for any given $q$ the set of all $(\gamma, \delta)$s where the largest achievable rate for $q$-ary $(\gamma,\delta)$-list decodable codes is non-zero. This fully resolved the zero-rate projection of the question addressed in this paper (shown in the bottom 2D chart of \Cref{fig:q=5-intro}).

\subsubsection{Alphabet dependent rate results for the deletion-only and insertion-only case}

The two other projections, i.e., bounds on the highest achievable rate for the insertion-only ($\delta = 0$) and deletion-only ($\gamma = 0$) cases in dependence on $q$ and the error parameter ($\gamma$ and $\delta$ respectively) were given by Haeupler et al.~\cite{haeupler2018synchronization4}. These projections are shown in \Cref{fig:q=5-intro} to the right and left respectively. The inner bounds presented in \cite{haeupler2018synchronization4} are derived by analyzing list-decoding properties of random codes. Here, we briefly review (the ideas of) the outer bounds from \cite{haeupler2018synchronization4} as these will be helpful for the remainder of this paper.

\paragraph{Deletion-only case.} A simple observation for deletion-only channels is that no family of positive-rate $q$-ary codes can be list-decoded from $\delta \geq 1-\frac{1}{q}$ fraction of deletions. This is due to a simple strategy that adversary can employ to eliminate all occurrences of all symbols of the alphabet except the most frequent one to convert any sent codeword into a word like $a^{n(1-\delta)}$ for some $a\in[q]$. \cite{haeupler2018synchronization4} suggests a similar strategy called \emph{Alphabet Reduction} for the adversary when $\delta = \frac{d}{q}$ for some integer $d$. With $\delta = \frac{d}{q}$ fraction of deletions, an adversary can remove all instances of the $d$ least frequent symbols and, hence, convert any transmitted codeword into a member of an ensemble of $(q-d)^{n(1-\delta)}$ strings. This implies an outer bound of $\frac{\log (q-d)^{n(1-\delta)}}{n\log q} = (1-\delta)\left(1-\log_q\frac{1}{1-\delta} \right)$ on the largest rate achievable by list-decodable deletion codes for special values of $\delta = \frac{d}{q}$ where $d = 1, 2, \cdots, q-1$. Using a simple time sharing argument between the alphabet reduction strategy over these points, \cite{haeupler2018synchronization4} provides a piece-wise linear outer bound for all values of $0 < \delta < 1-\frac{1}{q}$.

\paragraph{Insertion-only case.} In an insertion channel, the received word contains the sent codeword as a subsequence. To provide an outer bound on the highest achievable rate by insertion codes, \cite{haeupler2018synchronization4} used the probabilistic method: For a given codeword $x\in[q]^n$, \cite{haeupler2018synchronization4} computes the probability of a random string $y\in[q]^{n(1+\gamma)}$ containing $x$ as a subsequence. Having this quantity, one can compute the expected number of codewords of a given code $C$ with rate $r$ that are contained in a random string $y\in[q]^{n(1+\gamma)}$. Note that if $r$ is so high that this expectation is exponentially large in terms of $n$, then, by linearity of expectation, there exists some string $\bar{y}\in[q]^{n(1+\gamma)}$ which contains exponentially many codewords of $C$ which is a contradiction to its list-decodability from $\gamma n$ insertions. This implies an outer bound for the communication rate which we describe in more details below.

\subsubsection{General Case}



Liu et al.~\cite{liu2019listarxiv} was the first and only other work studying the rate of list-decodable insertion-deletion codes in  full generality, like this paper. After direct contradictions between the results reported here and the claims in \cite{liu2019listarxiv} were discovered, several correctness issues with key approaches of~\cite{liu2019listarxiv} for outer bounds were identified. These  results have been removed in \cite{liu2019list}, the final version of \cite{liu2019listarxiv}. The underlying issues seem hard to fix without substantially new ideas, as also reported in the acknowledgements of \cite{liu2019list}. As a result, \cite{liu2019list} is less directly relevant to this work, with the largest overlap being the inner bounds, similarly derived via a simple analysis of random insertion-deletion codes. Our bound is stronger for all pairs $(\gamma, \delta)$ when $q\geq 3$. \ITWonly{(See the full version for a proof.)}



\section{Outer Bound\noITW{s}}\label{sec:outer bound}
\noITW{
\subsection{Linear Outer Bounds from Resilience Results}

We start this section by providing a simple outer bound for best possible rate of an insdel code by generalizing the tight results for the resilience region of \cite{guruswami2019optimally} into a rate bound for any $(\gamma,\delta)$. Recall that the resilience region $F_q$ for any (integer) alphabet size $q >2$ is defined as the set of error rates for which there exists list-decodable codes with positive rate\noITW{, i.e.,
$$F_q = \left\{ (\gamma,\delta) \Big| 0 \leq \delta \leq \frac{q-1}{q}, 0 \leq \gamma \leq q-1, 
\exists \  (\gamma,\delta)\text{-list dec. $q$-ary code family with positive rate}\right\}.$$
In \cite{guruswami2019optimally}, the following exact description of $F_q$ was given:

\begin{theorem}[Theorem 1.3 of \cite{guruswami2019optimally}]\label{thm:feasibility-region}
For any positive integer $q \geq 2$, the resilience region $F_q$ is exactly the concave polygon defined over vertices $\left(\frac{i(i-1)}{q}, \frac{q-i}{q}\right)$ for $i = 1, \cdots, q$ and $(0, 0)$, not including the borders except the two segments
$\left[(0, 0), (q-1, 0)\right)$ and $\left[(0, 0), \left(0, 1-1/q\right)\right)$.\noITW{ (See \Cref{fig:actual-region})}\\
In particular, for any $\eps > 0$ and any $(\gamma, \delta) \in (1-\eps)F_q$, there exists a family of $q$-ary codes with positive rate that is $(\gamma, \delta)$-list-decodable. Further, for any $(\gamma, \delta) \not\in F_q$ there exists no $(\gamma, \delta)$-list-decodable family of $q$-ary codes with positive rate.



\end{theorem}

\global\def\FigureResilience{
\begin{figure}[]
  \centering
  \includegraphics[height=2in]{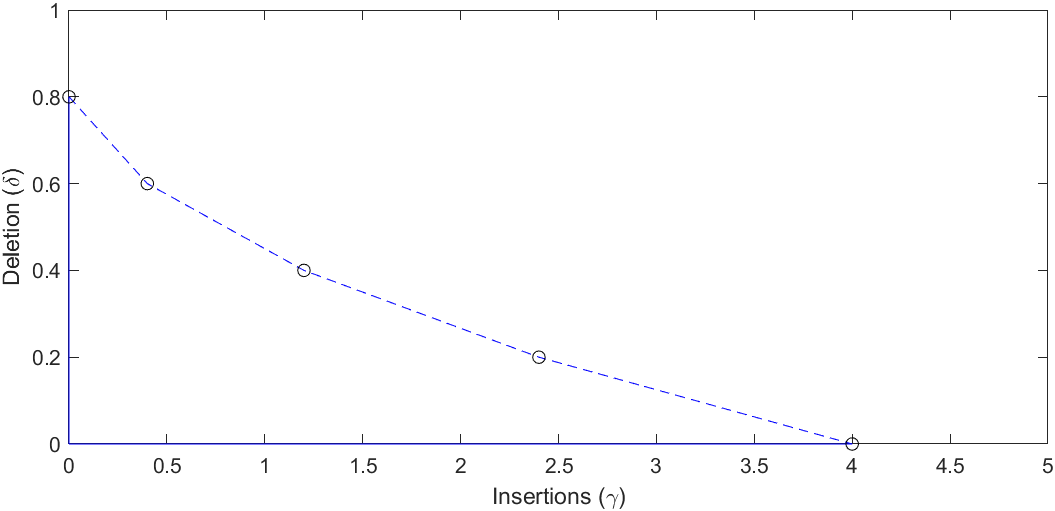}
  \caption{Resilience region for $q=5$.}\label{fig:actual-region}
\end{figure}
}
\FigureResilience
}\ITWonly{. See Theorem 1.3 of \cite{guruswami2019optimally} for an exact characterization of $F_q$.}

We show that the multiplicative distance to this resilience region gives a valid outer bound on the rate of any list-decodable insertion-deletion code:

\begin{theorem}\label{thm:convex-hull-outer bound}
For any alphabet size $q$ and any $(\gamma,\delta) \in F_q$ let $\alpha \geq 1$ be the smallest number\footnote{such minimum exists due to the definition of $F_q$.} such that $(\alpha \gamma,\alpha\delta) \notin F_q$. Any family of $(\gamma,\delta)$-list decodable $q$-ary codes cannot achieve a rate of more than $1 - 1/\alpha$.
\end{theorem}

\noITW{
A different way of looking at this outer bound is to think of it as the collection of lines that connect every point on the (border of the) resilience region $F_q$ identified in \cref{thm:feasibility-region} on the $r=0$ plane and the point $(\gamma,\delta, r) = (0,0, 1)$ which indicates the trivial achievable rate of 1 in the absence of noise. (See \Cref{fig:outer bound-with-resilience}) 

\global\def\FigureLinearOuterBound{
\begin{figure}[]
  \centering
  \includegraphics[scale=.7]{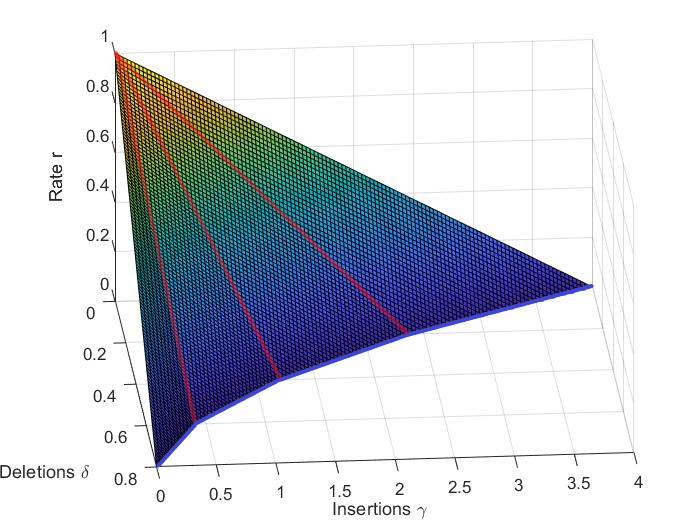}
  \caption{Illustration of the outer bound from \cref{thm:convex-hull-outer bound} for $q=5$.}\label{fig:outer bound-with-resilience}
\end{figure}
}
\FigureLinearOuterBound
}

\ITWonly{We refer the reader to the full version of this paper~\cite{FullVersion} for a formal proof but remark that the}
\noITW{The} proof of \Cref{thm:convex-hull-outer bound} is easy once one recalls how the outer bound for the feasibility region $F_q$ is proven in \cite{guruswami2019optimally}. It basically consists of a simple strategy transforming any sent string into one of a small $O_q(1)$ number of canonical strings, thus erasing almost all information sent. One can prove \Cref{thm:convex-hull-outer bound} by doing the same but only on an $\frac{1}{\alpha}$ fraction of the string. 
\noITW{Here though, for the sake of brevity and completeness, we present an alternative and shorter formal proof in the following.

\begin{proof}[Proof of \Cref{thm:convex-hull-outer bound}]
Assume for the sake of contradiction that for some $(\gamma,\delta) \in F_q$, for which $(\alpha \gamma,\alpha\delta) \notin F_q$, there exists a family of $(\gamma, \delta)$-list-decodable codes $\mathcal{C}$ with a higher rate than $1-\frac{1}{\alpha}$, i.e., codes $\mathcal{C} = \{C_1, C_2, \cdots\}$ with block lengths $n_1 < n_2 < \cdots$ and rates $r_1, r_2, \cdots$ that satisfy $r = \lim_{i\rightarrow \infty} r_i = 1-\frac{1}{\alpha}+\eps$ for some $\eps > 0$.

We convert this family of codes to a new family of codes $\mathcal{C}'$ by converting each code $C_i$ into a code $C'_i$ as follows: In all codewords of $C_i$, consider the $n_i\left(1-1/\alpha\right)$-long prefix. Among all such prefixes, let $p$ be the most frequent one. We set $C'_i$ to be a code containing all codewords of $C_i$ that start with $p$. Since all such codewords start with $p$, we omit the prefix $p$ from all such codewords. Note that the block length of $C'_i$ is $n'_i = n_i - n_i \left(1-1/\alpha\right) = n_i/\alpha$. Also, since there are $q^{n_i\left(1-1/\alpha\right)}$ $q$-ary strings of length $n_i\left(1-1/\alpha\right)$, 
\fullOnly{$$|C'_i| \geq \frac{|C_i|}{ q^{n_i\left(1-1/\alpha\right)} }=\frac{q^{n_i r_i}}{ q^{n_i\left(1-1/\alpha\right)} } = q^{n_i (r_i-1+1/\alpha)}.$$}
\shortOnly{$|C'_i| \geq \frac{|C_i|}{ q^{n_i\left(1-1/\alpha\right)} }=\frac{q^{n_i r_i}}{ q^{n_i\left(1-1/\alpha\right)} } = q^{n_i (r_i-1+1/\alpha)}.$}
This implies that the rate of $C'_i$ is at least 
\fullOnly{$$r'_i= \frac{\log_q |C'_i|}{n'_i} \geq \frac{n_i (r_i-1+1/\alpha)}{n'_i} = \alpha(r_i-1+1/\alpha)$$}
\shortOnly{$r'_i= \frac{\log_q |C'_i|}{n'_i} \geq \frac{n_i (r_i-1+1/\alpha)}{n'_i} = \alpha(r_i-1+1/\alpha)$}
and, hence, the rate of the family of codes $\mathcal{C}'$ is at least $\lim_{i\rightarrow\infty} r'_i \geq \alpha\eps > 0$.

Further, we claim that if $\mathcal{C}$ is $(\gamma, \delta, L(n))$-list decodable, then $\mathcal{C}'$ will be $(\alpha\gamma, \alpha\delta, L(\alpha n'))$-list decodable. To show this, we construct such list-decoder for all codes $C'_i \in \mathcal{C}'$ with input $y'$ by simply padding the most frequent $n_i(1-1/\alpha)$-prefix of codewords of $C_i \in \mathcal{C}$, $p$, in front of $y'$ and running the list-decoder of $C_i$ with input $y = p\cdot y'$. Among the list generated by the decoder of $C_i$, the ones that do not start with $p$ are withdrawn. The remaining strings will form the output of our list-decoder for $C'_i$ after omitting their prefix $p$. Note that this indeed gives a $(\alpha\gamma, \alpha\delta, L(\alpha n'))$-list-decoder since for any codeword $x\in C'_i$ that is $(\alpha\gamma, \alpha\delta)$-close to $y$, $p\cdot x \in C_i$ is $(\gamma, \delta)$-close to $p\cdot y$.

We were able to show that the family of codes $\mathcal{C}'$ achieves a positive rate and is $(\alpha\gamma, \alpha\delta)$-list decodable for $(\alpha\gamma, \alpha\delta) \notin F_q$. This is a contradiction to \cref{thm:feasibility-region} proving that the rate of $\mathcal{C}'$ may not exceed $1-\frac{1}{\alpha}$, thus, proving the theorem.
\end{proof}
}
We \ITWonly{also }remark that one can interpret the outer bound from \Cref{thm:convex-hull-outer bound} as a convexity argument in the following manner: We know that no code can achieve a rate of one in the presence of even a small amount of noise. The point $(\gamma, \delta, r) = (0,0,1)$ is therefore part of the (in)feasibility boundary. Further, the resilience result from \cite{guruswami2019optimally} demonstrates that all points within the region $R = \{(\gamma, \delta, 0)\ |\ (\gamma, \delta) \not\in F_q\}$ are infeasible. \Cref{thm:convex-hull-outer bound} shows that any convex combination of $(0,0,1)$ and any point in $R$ is infeasible as well, implying a pyramid-shaped feasibility region with $F_q$ as its base and $(0, 0, 1)$ as its apex. \noITW{(See \Cref{fig:outer bound-with-resilience}.)} 
We remark that generally convex combinations of infeasible points in the rate-distance tradeoff are not known to be infeasible - even for much simpler settings including Hamming errors or unique-decoding. Even convexity results between known infeasible points can be quite challenging. Case in point, our tighter outer bound in \cref{sec:tight-outer-bound} is proven by showing infeasibility of convex combinations of (easier) infeasible points.

\subsection{Stronger Bounds Using Generalizations of Bounds from \cite{haeupler2018synchronization4}}\label{sec:tight-outer-bound}

While \Cref{thm:convex-hull-outer bound} gives a bound for all $(\gamma,\delta)$, it is easy to see that it can be quite far from guarantees given by other state-of-the-art outer bounds. This is especially apparent for the deletion-only and insertion-only cases where \Cref{thm:convex-hull-outer bound} implies fairly trivial rate bounds of $1 - \frac{\delta}{1-1/q}$ and $1 - \frac{\gamma}{q-1}$, respectively. In particular, for the insertion-only case the following outer bound of \cite{haeupler2018synchronization4} gives much tighter bounds than $1 - \frac{\gamma}{q-1}$:}
\ITWonly{We start by reminding the following outer bound for the insertion-only case from \cite{haeupler2018synchronization4}.}

\begin{theorem}[From \cite{haeupler2018synchronization4}]\label{thm:converse}
 For any alphabet size $q$ and error rate $\gamma < q-1$, any family of $q$-ary codes $\mathcal{C}$ which is list-decodable from a $\gamma$ fraction of insertions has a rate of no more than $1- \log_q(\gamma+1)-\gamma\left(\log_q\frac{\gamma+1}{\gamma} - \log_q \frac{q}{q-1}\right)$.
\end{theorem}

Next, we show how to use \Cref{thm:converse} in a black-box fashion to give a very clean and easily statable outer bound for settings with both insertions and deletions, but in which the fraction of deletions has a nice form\noITW{, in particular, is a multiple of $\frac{1}{q}$}.\ITWonly{ (a multiple of $\frac{1}{q}$)} This outer bound forms the backbone of our final result.

\begin{theorem}\label{thm:nicedelta}
 For any fixed alphabet size $q$, any insertion rate $\gamma < q-1$ and any deletion rate $\delta = \frac{d}{q}$ for some integer $d < q$, it is true that any family of $q$-ary codes $\mathcal{C}$ which is $(\gamma,\delta)$-list-decodable has a rate of at most $(1-\delta)\left[
\left(1 + \frac{\gamma}{1-\delta}\right)\log_q\frac{q-d}{\frac{\gamma}{1-\delta}+1}
-\frac{\gamma}{1-\delta}\cdot\left(\log_{q}\frac{q-d-1}{\frac{\gamma}{1-\delta}}\right)\right]$.
\end{theorem}
\begin{proof}
Consider a code $C$ that is $(\gamma,\delta)$-list-decodable and assume that $\delta = \frac{d}{q}$ for some integer $d$. 
Assume that we restrict the adversary to utilize its deletions in the following manner: The adversary uses the $\frac{d}{q}$ deletion to remove all occurrences of the $d$-least frequent symbols of the alphabet. If there are remaining deletions, the adversary removes symbols from the end of the transmitted word.

Let us define the code $C'$ that is obtained from $C$ by deleting a $\delta$ fraction of symbols from each codeword of $C$ as described above. Note that the block length of $C'$ is $n'=n(1-\delta)$ and each of its codewords consist of up to $q'=q(1-\delta) = q-d$ symbols of the alphabet though this subset of size $q-d$ may be different from codeword to codeword. We partition the codewords of $C'$ into $q \choose q-d$ sets $C'_1, C'_2, \cdots, C'_{q \choose q-d}$ based on which $(q-d)$-subset of the alphabet they consist of.

Since $C$ is $(\gamma,\delta)$-list-decodable, each of the $C'_i$s are list-decodable from $\gamma n$ insertions. Therefore, \cref{thm:converse} implies that the size of each code $C'_i$ is no larger than
\noITW{$$q'^{n'\left[1- \log_{q'}(\gamma'+1)-\gamma'\left(\log_{q'}\frac{\gamma'+1}{\gamma'} - \log_{q'} \frac{q'}{q'-1}\right)\right]}$$}
\ITWonly{$q'^{n'\left[1- \log_{q'}(\gamma'+1)-\gamma'\left(\log_{q'}\frac{\gamma'+1}{\gamma'} - \log_{q'} \frac{q'}{q'-1}\right)\right]}$}
where $q'=q-d$, $n'=n(1-\delta)$, and $\gamma' = \frac{\gamma}{1-\delta}$. Therefore, the size of the code $C$ is no larger than
\noITW{
\fullOnly{$${q\choose q-d} q^{n(1-\delta)\left[\log_q q'- \log_{q}(\gamma'+1)-\gamma'\left(\log_{q}\frac{\gamma'+1}{\gamma'} - \log_{q} \frac{q'}{q'-1}\right)\right]}$$}
\shortOnly{${q\choose q-d} q^{n(1-\delta)\left[\log_q q'- \log_{q}(\gamma'+1)-\gamma'\left(\log_{q}\frac{\gamma'+1}{\gamma'} - \log_{q} \frac{q'}{q'-1}\right)\right]}$}
}\ITWonly{${q\choose q-d} q^{n(1-\delta)\left[\log_q q'- \log_{q}(\gamma'+1)-\gamma'\left(\log_{q}\frac{\gamma'+1}{\gamma'} - \log_{q} \frac{q'}{q'-1}\right)\right]}$ which implies the theorem statement.}
\noITW{and, consequently, its rate is no larger than
\begin{eqnarray*}
&(1-\delta)\left[\log_q (q-d)- \log_{q}\left(\frac{\gamma}{1-\delta}+1\right)-\frac{\gamma}{1-\delta}\cdot\left(\log_{q}\frac{\gamma+1-\delta}{\gamma} - \log_{q} \frac{q-d}{q-d-1}\right)\right]\\
=&(1-\delta)\left[
\left(1 + \frac{\gamma}{1-\delta}\right)\log_q\frac{q-d}{\frac{\gamma}{1-\delta}+1}
-\frac{\gamma}{1-\delta}\cdot\left(\log_{q}\frac{q-d-1}{\frac{\gamma}{1-\delta}}\right)
\right]\label{eqn:nice-delta-outer bound}.
\end{eqnarray*}
}
\ITWonly{
}
\end{proof}

Given the nice and explicit form of \Cref{thm:nicedelta} for any $q$ and $\gamma$ with multiple specific values of $\delta$, it seems tempting to conjecture that the restriction of $\delta$ is unnecessary making
$(1-\delta)\left[
\left(1 + \frac{\gamma}{1-\delta}\right)\log_q\frac{q-d}{\frac{\gamma}{1-\delta}+1}
-\frac{\gamma}{1-\delta}\cdot\left(\log_{q}\frac{q-d-1}{\frac{\gamma}{1-\delta}}\right)\right]$ a valid outer bound for any value of $\delta$ (and $\gamma$). This, however, could not be further from the truth. Indeed, for any $\delta$ not of the form restricted to by \Cref{thm:nicedelta}, there exists a $\gamma$ for which this extended bound is provably wrong because it contradicts the existence of the list-decodable codes constructed in \cite{guruswami2019optimally}.

In fact, for the valid points where $\delta$ is a multiple of $\frac{1}{q}$, the rate bound of \Cref{thm:nicedelta} hits zero at exactly the corner points of the piece-wise linear resilience region $F_q$ characterized by \noITW{\Cref{thm:feasibility-region}}\ITWonly{\cite{guruswami2019optimally}}. Taking this as an inspiration, one could try to extend the bound of \Cref{thm:nicedelta} to all values of $\delta$ by considering for each $q$ and each rate $r$ the roughly $\frac{q}{r}$ points where \Cref{thm:nicedelta} hits the  plane corresponding to rate $r$ and extend these points in a piece-wise linear manner to a complete 2D-curve for this rate $r$. This would give a rate bound for any $\gamma,\delta$, and $q$ as desired, which reduces to a piece-wise linear function for any fixed $r$ and \noITW{also} correctly reproduce $F_q$ for $r=0$. 

It turns out that this is indeed a correct outer bound. However, a stronger form of convexity, which takes full 3D-convex interpolations between any points supplied by \Cref{thm:nicedelta} and in particular combines points with different rates, also holds and is needed to give our final outer bound. 


\begin{theorem}\label{thm:mainconvexitystatement}
For a fixed $q$, suppose that $(\gamma_0, \delta_0=\frac{d_0}{q})$ and $(\gamma_1, \delta_1=\frac{d_1}{q})$ are two error rate combinations for which \Cref{thm:nicedelta} implies a maximal communication rate of $r_0$ and $r_1$, respectively. For any $0 \leq \alpha \leq 1$ consider the following convex combinations of these quantities: $\gamma = \alpha\gamma_0+(1-\alpha)\gamma_1$, $\delta = \alpha\delta_0+(1-\alpha)\delta_1$, and $r = \alpha r_0 + (1-\alpha)r_1$.
It is true that any $(\delta,\gamma)$-list-decodable $q$-ary code \noITW{$C$} has a rate of at most $r$.
\end{theorem}
See \Cref{fig:q=2} for an illustration of this bound for $q=5$. Red curves indicate the outer bound described above for the special values of $\delta$ of the form $\frac{d}{q}$ as given by \Cref{thm:nicedelta}.

\Cref{thm:mainconvexitystatement} together with \Cref{thm:nicedelta} gives a conceptually very clean description of our outer bound. However, an (exact) evaluation of the outer bound as given by \Cref{thm:mainconvexitystatement} is not straightforward since there are many convex combinations which all produce valid bound but how to compute or select the one which gives the strongest guarantee on the rate for a given $(\gamma,\delta)$ pair is not clear. This is particularly true since, as already mentioned above, the optimal points to combine do not lie on the same rate-plane. To remedy this, we give, as an alternative statement to  \Cref{thm:mainconvexitystatement}, the next theorem which produces an explicit outer bound for any $(\gamma, \delta)$ as an $\alpha$-convex combination of two points $(\gamma_0, \delta_0)$ and $(\gamma_1, \delta_1)$ only in dependence on the free parameter $\gamma_0$. We then show in \Cref{lem:optimal-gamma} an explicit expression for the optimal value for $\gamma_0$. Together, this produces a significantly less clean but on the other hand fully explicit description of our outer bound. 



\global\def\FigureNiceDeltas{
\noITW{\begin{sidewaysfigure}
    \centering
    \includegraphics[width=1\textwidth]{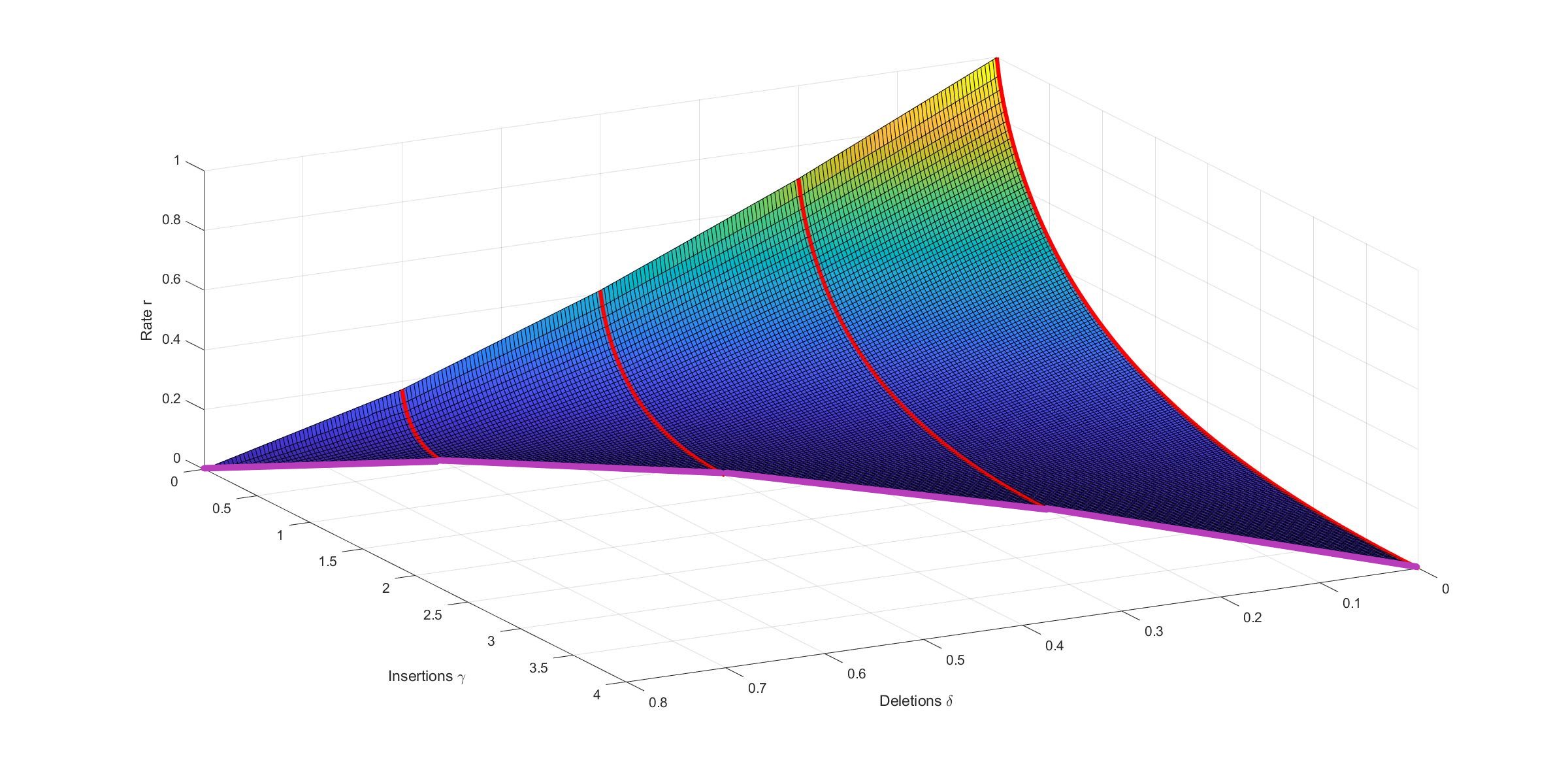}
    \caption{Outer bound for rate for $q=5$. The special case where $\delta=\frac{d}{q}$ for some integer $d$ is indicated with red lines.}
    \label{fig:q=2}
\end{sidewaysfigure}
}}
\ITWonly{
\begin{figure}
    \centering
    \includegraphics[width=1\linewidth]{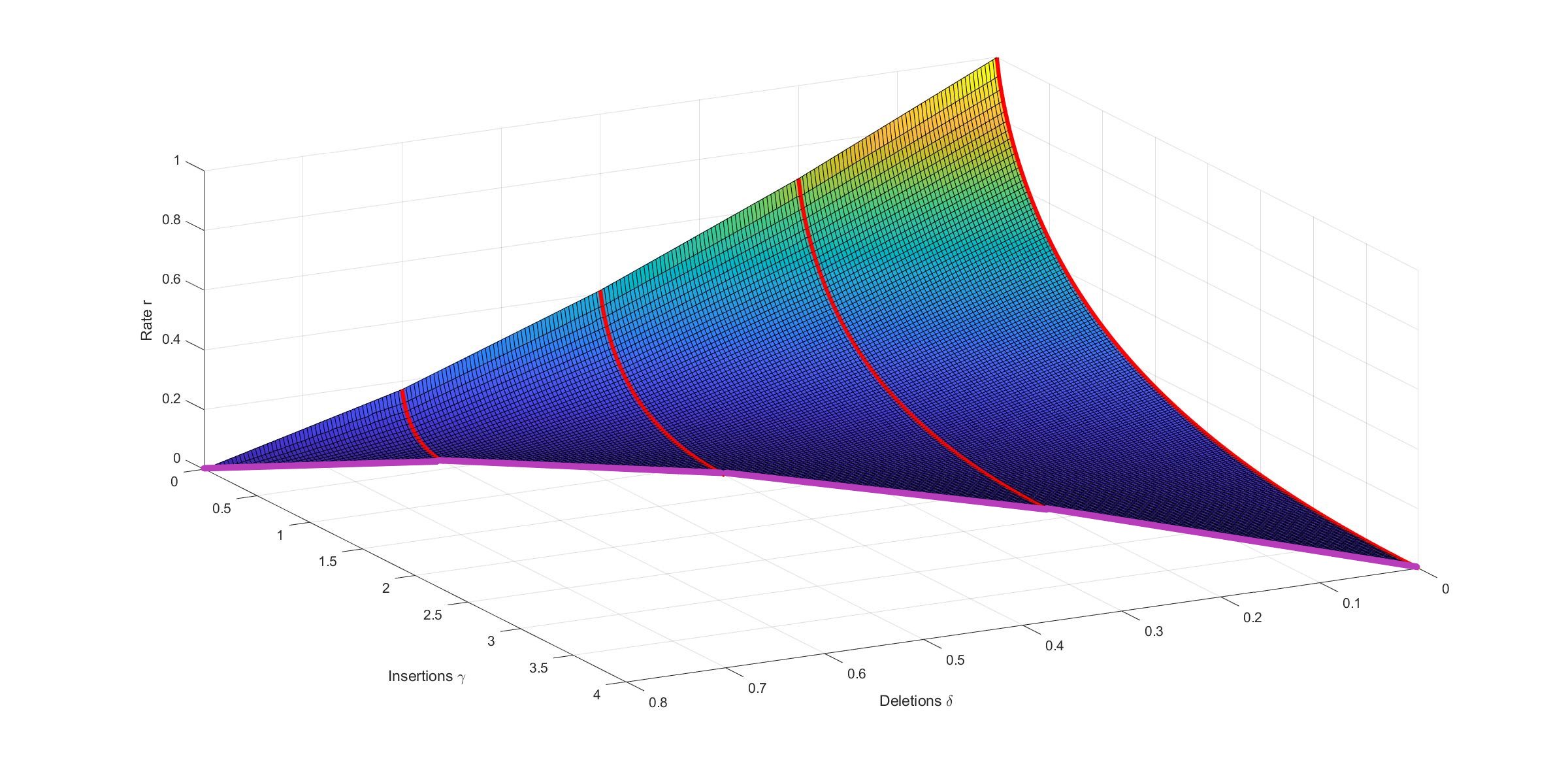}
    \caption{\footnotesize Outer bound for rate for $q=5$. The special case where $\delta=\frac{d}{q}$ for some integer $d$ is indicated with red lines.}
    \label{fig:q=2}
\end{figure}
}

\fullOnly{\FigureNiceDeltas}


\begin{theorem}\label{thm:outer bound}
Let $C$ be a $q$-ary insertion-deletion code that is list-decodable from $\gamma \in [0, q-1]$ fraction of insertions and $\delta \in [0, 1-\frac{1}{q}]$ fraction of deletions. Then, the rate of $C$ is no larger than
\noITW{\begin{eqnarray*}
&&    \alpha\left(1-\frac{d}{q}\right) \left(
(1+\gamma_0)\log_{q} \frac{q-d}{1+\gamma_0}
-
\gamma_0 \log_{q}\frac{q-d - 1}{\gamma_0}
\right)\allowdisplaybreaks\\
&+&
(1-\alpha)\left(1-\frac{d+1}{q}\right)
\left(
(1+\gamma_1)\log_{q} \frac{q-d-1}{1+\gamma_1}
-
\gamma_1 \log_{q}\frac{q-d - 2}{\gamma_1}
\right)
\end{eqnarray*}}
\ITWonly{
$\alpha\left(1-\frac{d}{q}\right) \left(
(1+\gamma_0)\log_{q} \frac{q-d}{1+\gamma_0}
-
\gamma_0 \log_{q}\frac{q-d - 1}{\gamma_0}
\right)+\\
(1-\alpha)\left(1-\frac{d+1}{q}\right) \left(
(1+\gamma_1)\log_{q} \frac{q-d-1}{1+\gamma_1}
-
\gamma_1 \log_{q}\frac{q-d - 2}{\gamma_1}
\right)$
}
for $d =\lfloor \delta q \rfloor$, $\alpha = 1 - \delta q + d$, and all $\gamma_0, \gamma_1 \geq 0$ where
$\alpha(1-\frac{d}{q})\gamma_0
+
(1-\alpha)(1-\frac{d+1}{q}) \gamma_1
 = \gamma$. We present the optimal choice of $\gamma_0$ in \cref{lem:optimal-gamma}.
\end{theorem}


\begin{proof}[Proof of \cref{thm:outer bound,thm:mainconvexitystatement}]

We first note that the statements of \cref{thm:outer bound} and \cref{thm:mainconvexitystatement} are merely a rephrasing of each other with the exception that \cref{thm:outer bound} only allows and optimizes over convex combinations of neighboring spokes of \Cref{thm:nicedelta}, namely the ones for $d_0 = d$ and $d_1 = d+1$ for $d = \lfloor\delta q\rfloor$. This restriction, however, is without loss of generality. Indeed, for any values from the domain $\left\{(\gamma, \delta) \mid \delta = \frac{d}{q}, 0\leq d \leq q-1, d\in \mathbb{Z}\right\}$, \Cref{thm:nicedelta} gives values which come from the function $f(\gamma, \delta) = (1-\delta)\left[
\left(1 + \frac{\gamma}{1-\delta}\right)\log_q\frac{q(1-\delta)}{\frac{\gamma}{1-\delta}+1}
-\frac{\gamma}{1-\delta}\cdot\left(\log_{q}\frac{q(1-\delta)-1}{\frac{\gamma}{1-\delta}}\right)
\right]$. This function is convex. \noITW{(see \noITW{\Cref{appdx:proofOfConvexity}} for a formal proof.)}\ITWonly{(proof in \cite{FullVersion}.)} Any value given as a convex combination between two non-neighboring spokes can therefore be at least matched \noITW{(and indeed, thanks to the strict convexity of $f(\cdot,\cdot)$ always be improved)}\ITWonly{(and actually improved due to strict convexity)} by choosing a different convex combination between neighboring spokes. This justifies the ``restricted" formulation of \cref{thm:outer bound}, which helps in reducing the number of parameters and simplifies calculations. 

In order to prove \cref{thm:outer bound} we, again, fix a  specific strategy for the adversary's use of deletions. In particular, the adversary will use $n\alpha\frac{d}{q}$ deletions on the first $n\alpha$ symbols of the transmitted codeword to eliminate all instances of the $d$ least-frequent symbols there. Similarly, he removes all instances of the respective $d+1$ least frequent symbols from the last $n(1-\alpha)$ symbols of the codeword. The resulting string is one out of some $\Sigma_0 ^ {n\alpha(1-\frac{d}{q})} \times \Sigma_1 ^ {n(1-\alpha)(1-\frac{d+1}{q})}$ where $\Sigma_0, \Sigma_1 \subseteq [q]$, $q_0 = |\Sigma_0| = q - d$, $q_1 = |\Sigma_1| = q - d - 1$. \noITW{This deletion strategy fits within the budgeted number of deletions since $\delta = \alpha \frac{d}{q} + (1-\alpha)\frac{d + 1}{q}$ for $d = \lfloor \delta q \rfloor$ and $\alpha = 1 - \delta q + d$.}

Note that while the adversary can convert any codeword of $C$ to a string of such form, the sub-alphabets $\Sigma_0$ and $\Sigma_1$ will likely be different between different codewords of $C$. Let $(\Sigma_0, \Sigma_1)$ be the pair of the most frequently reduced to alphabets and let $C_0$ be the set of codewords of $C$ that, after undergoing the above-described procedure, turn into a string out of $\Sigma_0 ^ {n\alpha(1-\frac{d}{q})} \times \Sigma_1 ^ {n(1-\alpha)(1-\frac{d+1}{q})}$. Note that $\frac{|C|}{{q\choose d}{q\choose d + 1}}\leq |C_0| \leq |C|$. Further, let $D_0$ be the set of codewords in $C_0$ after undergoing the alphabet reduction procedure mentioned above. To give an outer bound of the rate of $C$ it thus suffices to bound from above the size of $C_0$--or equivalently, $D_0$; Since $C$ is ($L=\poly(n))$-list-decodable, no more than $L$ members of $C_0$ can be mapped to a single member of $D_0$\ITWonly{ thus $|D_0| \geq \frac{|C_0|}{\poly(n)}$.}\noITW{ during the above procedure, thus, $|D_0| \geq \frac{|C_0|}{\poly(n)}$ which results in a negligible $o(1)$ difference in the rate.}

We bound above the size of $D_0$ by showing that if $|D_0|$ is too large, there will be some received word that can be obtained by exponentially many words in $D_0$ after $n\gamma$ insertions. Similar to \cite{haeupler2018synchronization4}, we utilize the linearity of expectation to derive this. Let us pick a random string $Z = (Z_0, Z_1)$ that consists of $n\alpha(1-\frac{d}{q})(1 + \gamma_0)$ symbols chosen uniformly out of $\Sigma_0$ (referred to by $Z_0$) and 
$n(1-\alpha)(1-\frac{d+1}{q})(1+\gamma_1)$ symbols uniformly chosen out of $\Sigma_1$ (referred to by $Z_1$). We have that 
$ 
\alpha(1-\frac{d}{q})\gamma_0
+
(1-\alpha)(1-\frac{d+1}{q}) \gamma_1
 = \gamma $. ($\gamma_0$ and $\gamma_1$ will be determined later.) We calculate the expected number of the members of $D_0$ that are subsequences of such string -- denoted by $X$. In the following, we \noITW{will often }describe members of $D_0$ like $y$ as the concatenation $(y_0, y_1)$ where 
 $|y_0| = n_0 = n\alpha(1-\frac{d}{q})$
 and
 $|y_1| = n_1 = n(1-\alpha)(1-\frac{d+1}{q})$.
\noITW{
\begin{eqnarray}
\mathbb{E}[X]
&=& \sum_{y=(y_0, y_1)\in{D_0}}\Pr\{y\text{ is a subsequence of }Z\}\nonumber\\
 &\geq& \sum_{y=(y_0, y_1)\in{D_0}}\Pr\{y_0\text{ is a subsequence of }Z_0\} \cdot \Pr\{y_1\text{ is a subsequence of }Z_1\} \nonumber\\
&=& \sum_{y=(y_0, y_1)\in{D_0}}\prod_{j=0, 1} \Pr\{y_j\text{ is a subsequence of }Z_j\} \nonumber\allowdisplaybreaks\\
&=& \sum_{y=(y_0, y_1)\in{D_0}}\prod_{j=0, 1} \sum_{1\le a_1<a_2<\cdots <a_{n_j} \le     n_j(1+\gamma_j)} \frac{1}{|\Sigma_j|^{n_j}}\left(1-\frac{1}{|\Sigma_j|}\right)^{a_{n_j}-{n_j}}
\label{eqn:ExpectationComputationStep1}\\
&=& |D_0|\prod_{j=0, 1}\left(|\Sigma_j|-1\right)^{-n_j}\sum_{l=n_j}^{n_j(1+\gamma_j)} {l-1 \choose n_j - 1} \left(\frac{|\Sigma_j|-1}{|\Sigma_j|}\right)^l\nonumber\\
&\ge& |D_0|\prod_{j=0, 1}\left(|\Sigma_j|-1\right)^{-n_j} {n_j(1+\gamma_j)-1 \choose n_j-1} \left(\frac{|\Sigma_j|-1}{|\Sigma_j|}\right)^{n_j(1+\gamma_j)}\label{eqn:ExpectationComputationStep2}\\
&=& |D_0|\prod_{j=0, 1}\left(|\Sigma_j|-1\right)^{-n_j} \frac{1}{1+\gamma_j}{n_j(1+\gamma_j) \choose n_j} \left(\frac{|\Sigma_j|-1}{|\Sigma_j|}\right)^{n_j(1+\gamma_j)}\nonumber\\
&=& |D_0| \prod_{j=0, 1}(|\Sigma_j|-1)^{n_j\gamma_j}|\Sigma_j|^{-n_j(1+\gamma_j)}2^{n_j(1+\gamma_j)H\left(\frac{1}{1+\gamma_j}\right)+o(n)}\nonumber\\
&=& |D_0| \prod_{j=0, 1} q^{n_j\left(
\gamma_j \log_{q}(q_j-1)
-(1+\gamma_j)\log_q q_{j}
+(1+\gamma_j)\log_{q} (1+\gamma_j)
-\gamma_j\log_{q} \gamma_j
+o(1)\right)}\nonumber\\
&=& |D_0| \prod_{j=0, 1} q^{n_j\left(
\gamma_j \log_{q}\frac{q_j - 1}{\gamma_j}
-(1+\gamma_j)\log_{q} \frac{q_j}{1+\gamma_j}+o(1)\right)}\nonumber\\
&=& |D_0| q^{\sum_{j=0, 1} n_j\left(
\gamma_j \log_{q}\frac{q_j - 1}{\gamma_j}
-(1+\gamma_j)\log_{q} \frac{q_j}{1+\gamma_j}+o(1)\right)}
\label{eqn:ExpectationComputationStep3}
\end{eqnarray}
}
\ITWonly{

We have $\mathbb{E}[X]
= \sum_{y=(y_0, y_1)\in{D_0}}\Pr\{y\text{ is a subseq. of }Z\}.$
Note that 
$\Pr\{y\text{ is a subseq. of }Z\}$ is not smaller than $\prod_{j=0, 1} \Pr\{y_j\text{ is a subseq. of }Z_j\}
$. Also, conditioning on the leftmost occurrence of $Z_j$ in $y_j$, we can write down the $\Pr\{y_j\text{ is a subseq. of }Z_j\}$ as 
$\sum_{1\le a_1<\cdots <a_{n_j} \le     n_j(1+\gamma_j)} \frac{1}{|\Sigma_j|^{n_j}}\left(1-\frac{1}{|\Sigma_j|}\right)^{a_{n_j}-{n_j}}$. We use this expression in \cite{FullVersion} to bound below $\mathbb{E}[X]$ by
$L=|D_0| q^{\sum_{j=0, 1} n_j\left(
\gamma_j \log_{q}\frac{q_j - 1}{\gamma_j}
-(1+\gamma_j)\log_{q} \frac{q_j}{1+\gamma_j}+o(1)\right)}$.
}
\noITW{
Step~\eqref{eqn:ExpectationComputationStep1} is obtained by conditioning the probability of $y_j$ being a subsequence of $Z_j$ over the leftmost occurrence of $y_j$ in $Z_j$ indicated by $a_1, a_2, \cdots, a_n$ as indices of $Z_j$ where the leftmost occurrence of $y_j$ is located. In that event, $Z_j[a_i]=y_j[i]$ and $y_j[i]$ cannot appear in $Z_j[a_{i-1}+1, a_{i}-1]$. Therefore, the probability of this event is $\left(\frac{1}{q_j}\right)^{n_j}\left(1-\frac{1}{q_j}\right)^{a_n-n_j}$. 
To verify Step~\eqref{eqn:ExpectationComputationStep2}, note that we substituted a summation of positive values with one single term among them. 

Finally, by \eqref{eqn:ExpectationComputationStep3}, there exists some realization of $Z$ to which at least 
$$|D_0| q^{\sum_{j=0, 1} n_j\left(
\gamma_j \log_{q}\frac{q_j - 1}{\gamma_j}
-(1+\gamma_j)\log_{q} \frac{q_j}{1+\gamma_j}+o(1)\right)}$$}
\ITWonly{This means that there exists some realization of $Z$ to which at least $L$}
codewords of $\mathcal{C}$ are subsequences. In order for $C$ to be list-decodable, this quantity needs to be sub-exponential. Therefore, 
\ITWonly{$ r_C = \frac{\log_q |D_0| + O(1)}{n} \leq 
\sum_{j=0, 1} \frac{n_j}{n}\left(
(1+\gamma_j)\log_{q} \frac{q_j}{1+\gamma_j}
-
\gamma_j \log_{q}\frac{q_j - 1}{\gamma_j}
\right)$. This leads to the upper bound stated in \Cref{thm:outer bound} for $r_C$. (see the full version~\cite{FullVersion} for a complete calculation and proof.)}
\noITW{
\begin{eqnarray}
 r_C &=& \frac{\log_q |D_0| + O(1)}{n} \leq 
\sum_{j=0, 1} \frac{n_j}{n}\left(
(1+\gamma_j)\log_{q} \frac{q_j}{1+\gamma_j}
-
\gamma_j \log_{q}\frac{q_j - 1}{\gamma_j}
\right)\nonumber\\
&=&
\alpha\left(1-\frac{d}{q}\right) \left(
(1+\gamma_0)\log_{q} \frac{q_0}{1+\gamma_0}
-
\gamma_0 \log_{q}\frac{q_0 - 1}{\gamma_0}
\right)\nonumber\\
&&+
(1-\alpha)\left(1-\frac{d+1}{q}\right)\left(
(1+\gamma_1)\log_{q} \frac{q_1}{1+\gamma_1}
-
\gamma_1 \log_{q}\frac{q_1 - 1}{\gamma_1}
\right)\nonumber\\
&=&
\alpha\left(1-\frac{d}{q}\right) \left(
(1+\gamma_0)\log_{q} \frac{q-d}{1+\gamma_0}
-
\gamma_0 \log_{q}\frac{q-d - 1}{\gamma_0}
\right)\nonumber\\
&&+
(1-\alpha)\left(1-\frac{d+1}{q}\right)\left(
(1+\gamma_1)\log_{q} \frac{q-d-1}{1+\gamma_1}
-
\gamma_1 \log_{q}\frac{q-d - 2}{\gamma_1}
\right)\label{eqn:parametrized-bound}
\end{eqnarray}

Note that \eqref{eqn:parametrized-bound} is an outer bound for the rate for all choices of $\gamma_0, \gamma_1 \geq 0$ where 
$\alpha(1-\frac{d}{q})\gamma_0
+
(1-\alpha)(1-\frac{d+1}{q}) \gamma_1
 = \gamma.$}
\end{proof}
\ITWonly{\vspace{-4mm}}
\begin{theorem}\label{lem:optimal-gamma}
The optimal choice for \noITW{the value of }$\gamma_0$\noITW{\footnote{Note that choosing $\gamma_0$ also determines a value for $\gamma_1$.}} in \cref{thm:outer bound} \noITW{(i.e., the one that yields the tightest bound)} satisfies
\noITW{
$$\left(1 + \frac{1}{\gamma_1}\right)\left(1-\frac{1}{q-d-1}\right)= \left(1 + \frac{1}{\gamma_0}\right)\left(1-\frac{1}{q-d}\right).$$
\smallskip
}
\ITWonly{
$\left(1 + 1/\gamma_1\right)\left(1-{1}/(q-d-1)\right)= \left(1 + 1/{\gamma_0}\right)\left(1-1/(q-d)\right).$
}
\noindent Together with the equation $\alpha(1-\frac{d}{q})\gamma_0 + (1-\alpha)(1-\frac{d+1}{q}) \gamma_1  = \gamma $, this gives \noITW{the following explicit expression for $\gamma_0$ in terms of $q$,$\gamma$, $d =\lfloor \delta q \rfloor$, and $\alpha = 1 - \delta q + d$:
$$\gamma_0=\frac{1}{2\alpha (q-d)} \cdot \left(A - \sqrt{B^2 + C}\right)$$
}
\ITWonly{an explicit expression for $\gamma_0$ in terms of $q$,$\gamma$, $d =\lfloor \delta q \rfloor$, and $\alpha = 1 - \delta q + d$ which can be found in the full version \cite{FullVersion}.
}
\noITW{for 
\smallskip
$A = 3 \alpha  d^2 q+d^2 q-3 \alpha  d q^2-2 d q^2+4 \alpha  d q+2 d q+\alpha  q^3-2 \alpha  q^2+q^3-2 q^2+\gamma  q+q -\alpha  d^3-2 \alpha  d^2$

\smallskip

$B= \alpha  d^3+2 \alpha  d^2-3 \alpha  d^2 q-d^2 q+3 \alpha  d q^2+2 d q^2-4 \alpha  d q-2 d q-\alpha  q^3+2 \alpha  q^2-q^3+2 q^2-\gamma  q-q$

\smallskip

$C = 4 (\alpha  q-\alpha  d) \left(\gamma  d^2 q-2 \gamma  d q^2+2 \gamma  d q+\gamma  q^3-2 \gamma  q^2\right)$}
\end{theorem}
\global\def\ProofOfOptimalGammaTheorem{
\begin{proof}
To find the optimal value for $\gamma_0$, we find the choice of $\gamma_0$ that minimizes \eqref{eqn:ExpectationComputationStep1}. To this end, we calculate the ratio between the values of \eqref{eqn:ExpectationComputationStep1} when $n\gamma$ insertions are distributed between two parts as $(n_0\gamma_0, n_1\gamma_1)$ and when distributed as $(n_0\gamma_0 + 1, n_1\gamma_1 - 1)$. 
We show that this ratio monotonically increases as one increases $\gamma_0$. Therefore, to find the value of $\gamma_0$ for which  \eqref{eqn:ExpectationComputationStep1} is minimized, one only needs to find the choice of $\gamma_0$ for which this ratio is equal to one. 
\begin{eqnarray*}
&&\frac{
{n_0(1+\gamma_0) \choose n_0} \left(\frac{|\Sigma_0|-1}{|\Sigma_0|}\right)^{n_0(1+\gamma_0)}
{n_1(1+\gamma_1) \choose n_1} \left(\frac{|\Sigma_1|-1}{|\Sigma_1|}\right)^{n_1(1+\gamma_1)}
}{
{n_0(1+\gamma_0)+1 \choose n_0} \left(\frac{|\Sigma_0|-1}{|\Sigma_0|}\right)^{n_0(1+\gamma_0)+1}
{n_1(1+\gamma_1)-1 \choose n_1} \left(\frac{|\Sigma_1|-1}{|\Sigma_1|}\right)^{n_1(1+\gamma_1)-1}
}\allowdisplaybreaks\\
&=&\frac{
{n_0(1+\gamma_0) \choose n_0} 
{n_1(1+\gamma_1) \choose n_1} \left(\frac{|\Sigma_1|-1}{|\Sigma_1|}\right)
}{
{n_0(1+\gamma_0)+1 \choose n_0} \left(\frac{|\Sigma_0|-1}{|\Sigma_0|}\right)
{n_1(1+\gamma_1)-1 \choose n_1}
}\\
&=&\frac{
{n_0(1+\gamma_0) \choose n_0} 
{n_1(1+\gamma_1) \choose n_1} 
}{
{n_0(1+\gamma_0)+1 \choose n_0} 
{n_1(1+\gamma_1)-1 \choose n_1}
}
\times \frac{|\Sigma_0| \times (|\Sigma_1|-1)}{(|\Sigma_0|-1) \times |\Sigma_1|}
\\
&=&\frac{
\frac{n_1(1+\gamma_1)}{n_1(1+\gamma_1) - n_1} 
}{
\frac{n_0(1+\gamma_0)+1} {n_0(1+\gamma_0)+1 - n_0}
}
\times \frac{|\Sigma_0| \times (|\Sigma_1|-1)}{(|\Sigma_0|-1) \times |\Sigma_1|}
\\
&=&\frac{
\frac{n_1(1+\gamma_1)}{n_1\gamma_1} 
}{
\frac{n_0(1+\gamma_0)+1} {n_0\gamma_0+1}
}
\times \frac{|\Sigma_0| \times (|\Sigma_1|-1)}{(|\Sigma_0|-1) \times |\Sigma_1|}
\\
&=&\frac{
1 + \frac{1}{\gamma_1}
}{
1+\frac{1}{\gamma_0+1/n_0}
}
\times 
\frac{1-\frac{1}{|\Sigma_1|}}{1-\frac{1}{|\Sigma_0|}}
\end{eqnarray*}
Note that as one increases $\gamma_0$, the numerator grows and the denominator becomes smaller--meaning that the overall value goes up. 
Therefore, the optimal choice of $\gamma_0$ would be one for which:
\begin{eqnarray*}
&& \frac{1 + \frac{1}{\gamma_1}
}{
1+\frac{1}{\gamma_0+1/n_0}
}
\times 
\frac{1-\frac{1}{|\Sigma_1|}}{1-\frac{1}{|\Sigma_0|}} = 1\\
&\Leftrightarrow&
\left(1 + \frac{1}{\gamma_1}\right)\left(1-\frac{1}{|\Sigma_1|}\right)=
\left(1+\frac{1}{\gamma_0+1/n_0}\right)\left(1-\frac{1}{|\Sigma_0|}\right)
\end{eqnarray*}
Given that families of codes with increasing block lengths $n$ are considered, the term $\frac{1}{n_0} = \frac{1}{n\cdot \alpha(1-d/q)}$ vanishes. Thus, we are looking for a choice of $\gamma_0$ that satisfies
$$
\left(1 + \frac{1}{\gamma_1}\right)\left(1-\frac{1}{|\Sigma_1|}\right)=
\left(1 + \frac{1}{\gamma_0}\right)\left(1-\frac{1}{|\Sigma_0|}\right).$$

Putting this together with the equation $\alpha(1-\frac{d}{q})\gamma_0 + (1-\alpha)(1-\frac{d+1}{q}) \gamma_1  = \gamma $ from \cref{thm:outer bound} and solving the resulting system of equations analytically using computer software, the stated equation for the optimal choice of $\gamma_0$ is derived.
\end{proof}
}
\noITW{\ProofOfOptimalGammaTheorem}

\section{Inner Bound via Analyzing Random Codes}\label{sec:inner_bound}
In this section, we provide an inner bound on the highest rate achievable by list-decodable insertion-deletion codes. \noITW{We start with a preliminary lemma in the following.} Throughout this section, we define $\mathcal{B}_i(S, n_i)$ or the \emph{insertion sphere of radius $n_i$} as the set of all strings that can be obtained by $n_i$ insertions from $S$. $\mathcal{B}_d(S, n_d)$ and $\mathcal{B}(S, n_i, n_d)$ are similarly defined\noITW{ as the \emph{deletion sphere of radius $n_d$} and the \emph{insertion-deletion sphere of insertion radius $n_i$ and deletion radius $n_d$} around $S$}\ITWonly{ for deletions and combination of insertions and deletions}.
\begin{lemma}[From \cite{levenshtein1974elements}]\label{lem:insertion-ball}
Let $n, n_i$, and $q$ be positive integers and $S\in[q]^{n}$. Then, \noITW{the size of the insertion sphere of radius $n_i$ around $S$ is
$$|\mathcal{B}_i(S, n_i)| = \sum_{i=0}^{n_i}{n+n_i \choose i}(q-1)^i.$$}
\ITWonly{$|\mathcal{B}_i(S, n_i)| = \sum_{i=0}^{n_i}{n+n_i \choose i}(q-1)^i.$}
\end{lemma}


\noITW{In the following, we give a simple bound on the size of the insertion-deletion sphere.}
\begin{lemma}\label{lem:insdel-ball}
Let $x\in[q]^n$, $\delta \in \left[0, 1-\frac{1}{q}\right]$ and $\gamma \in [0, (q-1)(1-\delta)]$. The size of the insertion-deletion sphere of insertion-radius $\gamma n$ and deletion-radius $\delta n$ around $x$ is no larger than
\noITW{$$|\mathcal{B}(x, \gamma n, \delta n)| \leq q^{n\left(H_q(\delta)
+ (1-\delta+\gamma)H_q\left(\frac{\gamma}{1-\delta+\gamma}\right) 
- \delta\log_q(q-1) 
\right)+o(n)}$$}
\ITWonly{$q^{n\left(H_q(\delta)
+ (1-\delta+\gamma)H_q\left(\frac{\gamma}{1-\delta+\gamma}\right) 
- \delta\log_q(q-1) 
\right)+o(n)}$}
where $H_q(\cdot)$ denotes the $q$-ary entropy function\noITW{ (see definition in \eqref{eqn:entropy})}\ITWonly{ defined as
$H_q(x) = x\log_q(q-1) - x\log_q x - (1-x) \log_q(1-x)$
}.
\end{lemma}
\ITWonly{
\begin{proof}[Proof Sketch (Full proof available in \cite{FullVersion})]
{\footnotesize
\begin{align}
&|\mathcal{B}(x, \gamma n, \delta n)| \leq \sum_{x_0 \in \mathcal{B}_d(x, \delta n)} |\mathcal{B}_i(x_0, \gamma n)|\nonumber\\
&\leq  {n \choose \delta n} \sum_{i=0}^{\gamma n}{n(1-\delta+\gamma) \choose i}(q-1)^i
\label{eqn:ball1}\\
&\leq n\gamma {n \choose n\delta} {n(1-\delta+\gamma) \choose n\gamma} (q-1)^{\gamma n}\label{eqn:ball2}\allowdisplaybreaks\\
&=q^{n\left(H_q(\delta) 
+ (1-\delta+\gamma)H_q\left(\frac{\gamma}{1-\delta+\gamma}\right) 
- \delta\log_q(q-1) 
\right)+o(n)}\nonumber
\end{align} }
Note that \eqref{eqn:ball1} follows from \cref{lem:insertion-ball} and \eqref{eqn:ball2} is true because the term in summation reaches its maximum when $i=n\gamma$.
\end{proof}
}
\noITW{
\begin{proof}
\begin{eqnarray}
|\mathcal{B}(x, \gamma n, \delta n)| &\leq& \sum_{x_0 \in \mathcal{B}_d(x, \delta n)} |\mathcal{B}_i(x_0, \gamma n)|\nonumber\\
&\leq&  {n \choose \delta n} \sum_{i=0}^{\gamma n}{n(1-\delta+\gamma) \choose i}(q-1)^i
\label{eqn:ball1}\\
&\leq& n\gamma {n \choose n\delta} {n(1-\delta+\gamma) \choose n\gamma} (q-1)^{\gamma n}\label{eqn:ball2}\\
&=&q^{n\left(H_q(\delta) - \delta\log_q(q-1)
+ (1-\delta+\gamma)H_q\left(\frac{\gamma}{1-\delta+\gamma}\right) - \gamma\log_q(q-1) + \gamma\log_q(q-1)
\right)+o(n)}\label{eqn:ball3}\\
&=&q^{n\left(H_q(\delta) 
+ (1-\delta+\gamma)H_q\left(\frac{\gamma}{1-\delta+\gamma}\right) 
- \delta\log_q(q-1) 
\right)+o(n)}\nonumber
\end{eqnarray}
Note that \eqref{eqn:ball1} follows from \cref{lem:insertion-ball} and \eqref{eqn:ball2} is true because the term in summation reaches its maximum when $i=n\gamma$. To see this, we test the ratio between the value of the term for two consecutive parameter values $i$ and $i+1$:
$$\frac{{n(1-\delta+\gamma) \choose i+1}(q-1)^{i+1}}{{n(1-\delta+\gamma) \choose i}(q-1)^i} = \frac{n(1-\delta+\gamma)-i}{i+1}(q-1)$$
Note that 
$\frac{n(1-\delta+\gamma)-i}{i+1}(q-1)\geq 1 \Leftrightarrow \frac{n(1-\delta+\gamma)-i}{i+1} \geq \frac{1}{q-1} \Leftrightarrow iq+1 \leq n(1-\delta+\gamma)(q-1)
$. This holds for all $i\leq n\gamma$ because:
$$n\gamma q+1 \leq n(1-\delta+\gamma)(q-1) \Leftrightarrow n\gamma < n(1-\delta)(q-1).$$
Finally, \eqref{eqn:ball3} follows from the definition of the $q$-qry entropy function
\begin{equation}
H_q(x) = x\log_q(q-1) - x\log_q x - (1-x) \log_q(1-x)\label{eqn:entropy}
\end{equation}
and the equation ${n\choose np} = q^{n(H_q(p)-p\log_q(q-1))+o(n)}$.
\end{proof}
}

Using the bound on the size of the insertion-deletion radius presented above, we give the following inner bound on the highest achievable rate for $(\gamma, \delta)$-list-decodable codes derived by analysis of the list-decodability of random codes.

\begin{theorem}\label{thm:inner-bound-main}
For any integer $q\geq 2$, $\delta \in \left[0, 1-\frac{1}{q}\right]$ and $\gamma \in [0, (q-1)(1-\delta)]$, a family of random $q$-ary codes with rate
\noITW{$$0 \leq R < 1- 
(1-\delta+\gamma)H_q\left(\frac{\gamma}{1-\delta+\gamma}\right)
- H_q\left(\delta\right) 
+ \gamma\log_q(q-1)
$$}
\ITWonly{$R < 1- 
(1-\delta+\gamma)H_q\left(\frac{\gamma}{1-\delta+\gamma}\right)
- H_q\left(\delta\right) 
+ \gamma\log_q(q-1)
$}
is list-decodable from $\gamma n$ insertions and $\delta n$ deletions\noITW{ with high probability.}\ITWonly{ WHP.}
\end{theorem}
\noITW{
We remark that the condition $\gamma \leq (q-1)(1-\delta)$ does not weaken the statement of the theorem since there is no positive-rate family of codes that can list-decode from $(\gamma, \delta)$ fraction of errors where $\gamma \geq (1-\delta)(q-1)$. To see this, similar to the proof of \cref{thm:outer bound}, one can think of an adversary that reduces the alphabet to one of size $q-d$ in the first $n\alpha$ symbols of the message and to $q-d-1$ in the rest of it where $d = \lfloor \delta q \rfloor$ and $\alpha = 1 - \delta q + d$. With $$\gamma n \geq n(1-\delta)(q-1) \geq n\alpha\left(1-\frac{d}{q}\right)(q-d-1) + n(1-\alpha)\left(1-\frac{d+1}{q}\right)(q-d-2)$$ insertions, the adversary can turn any sent message into a string out of an ensemble of ${q\choose q-d}{q\choose q-d-1}$ strings by turning each message into repetitions of the reduce alphabet members as described in the proof of \cref{thm:outer bound}.
}
\ITWonly{
\begin{proof}[Proof Sketch (Full proof in \cite{FullVersion})]
We use \cref{lem:insdel-ball} to bound the probability of a fixed string falling within a certain ball around a codeword of a random code. We then bound above the probability of such \noITW{fixed} string being close to $l+1$ codewords, i.e., violating the $l$-list-decodability condition. Taking an upper bound over all such center strings, we bound above the probability of a random code not being list-decodable and find a range for $R$ where such probability is negligible.
\end{proof}
}
\noITW{
\begin{proof}[Proof of \cref{thm:inner-bound-main}]
Take the random codeword $X\in[q]^{n}$ and some string $y\in[q]^{n(1-\delta+\gamma)}$ of length $n'=n(1-\delta+\gamma)$. Using \cref{lem:insdel-ball}, the probability of $X$ being inside the insertion-deletion sphere of deletion-radius $\delta'n' = \gamma n$ and insertion-radius $\gamma'n' = \delta n$ of $y$ is
\begin{eqnarray*}
\Pr\{X\in\mathcal{B}(y, \gamma' n', \delta' n')\} 
&\leq& \frac{
q^{n'\left(H_q(\delta')
+ (1-\delta'+\gamma')H_q\left(\frac{\gamma'}{1-\delta'+\gamma'}\right) 
- \delta'\log_q(q-1) 
\right)+o(n)}
}{q^n}\\
&=& \frac{
q^{n(1-\delta+\gamma)\left(H_q(\frac{\gamma}{1-\delta+\gamma})
+ \frac{1}{1-\delta+\gamma}H_q\left(\delta\right) 
- \frac{\gamma}{1-\delta+\gamma}\log_q(q-1)
\right)+o(n)}
}{q^n}\\
&=& q^{n\left(H_q(\delta)
+ (1-\delta+\gamma)H_q\left(\frac{\gamma}{1-\delta+\gamma}\right) 
- \gamma\log_q(q-1) 
 - 1\right)+o(n)}
\end{eqnarray*}
For the random code $C$ with rate $R$ to not be $l$-list decodable for some integer $l$, there has to exists some string $y\in[q]^{n(1-\delta+\gamma)}$ that can be obtained by $l+1$ codewords of $C$ via $\delta n$ deletions and $\gamma n$ insertions, i.e., codewords that lie in $\mathcal{B}(y, \delta n, \gamma n)$. Using the union bound over all $y\in[q]^{n(1-\delta+\gamma)}$, the probability of the existence of such $y$ is at most.
 \begin{eqnarray}
&&q^{n(1-\delta+\gamma)} \left(q^{nR}\right)^{l+1} \left(q^{n\left(H_q(\delta)
+ (1-\delta+\gamma)H_q\left(\frac{\gamma}{1-\delta+\gamma}\right) 
- \gamma\log_q(q-1) 
 - 1\right)+o(n)}\right)^{l+1}\nonumber\\
&=&q^{n(l+1)
\left(R + H_q(\delta)
+ (1-\delta+\gamma)H_q\left(\frac{\gamma}{1-\delta+\gamma}\right) 
- \gamma\log_q(q-1) 
 - 1  + \frac{1-\delta+\gamma}{l+1}+o_l(1)
\right)
 }\label{eqn:list-dec-prob}
 \end{eqnarray}
 Note that we used the trivial bound ${q^{nR} \choose l+1} \leq \left(q^{nR}\right)^{l+1}$ in the above calculation.
Equation \eqref{eqn:list-dec-prob} implies that as long as 
 $$R <  1- 
 (1-\delta+\gamma)H_q\left(\frac{\gamma}{1-\delta+\gamma}\right)
 - H_q\left(\delta\right) 
+ \gamma\log_q(q-1),
 $$
for an appropriately large $l=O_{\gamma,\delta,q}(1)$, the exponent of \eqref{eqn:list-dec-prob} is negative and, therefore, the family of random codes is $(\gamma, \delta)$-list-decodable with high probability.
\end{proof}
}

\noITW{
We next show that this bound is stronger than the one presented by \cite{liu2019list} for the case of alphabet size $q\geq 3$.
\begin{theorem}\label{thm:compare-lower-bounds-q3}
The bound presented in \cref{thm:inner-bound-main} is stronger (i.e. larger) than the following that is provided in Lemma 15 of \cite{liu2019list} for any $\delta$ and $\gamma$.
$$1- 
 (1-\delta+2\gamma)H_q\left(\frac{\gamma}{1-\delta+2\gamma}\right)
 - H_q\left(\delta\right) 
+ \gamma\log_q(q-1)$$
\end{theorem}
\begin{proof}
To prove this, one simply needs to verify that 
\begin{eqnarray*}
&&(1-\delta+\gamma)H_q\left(\frac{\gamma}{1-\delta+\gamma}\right)
\leq
(1-\delta+2\gamma)H_q\left(\frac{\gamma}{1-\delta+2\gamma}\right)\\
&\Leftrightarrow&
-\gamma \log_q\left(\frac{\gamma}{1-\delta+\gamma}\right)
-(1-\delta)\log_q\left(\frac{1-\delta}{1-\delta+\gamma}\right)
+\gamma \log_q(q-1)\\
&&\leq -\gamma \log_q\left(\frac{\gamma}{1-\delta+2\gamma}\right)
-(1-\delta+\gamma)\log_q\left(\frac{1-\delta+\gamma}{1-\delta+2\gamma}\right)+\gamma \log_q(q-1)\\
&\Leftrightarrow&
-\gamma \log_q\left(\frac{\gamma}{1-\delta+\gamma}\right)
-(1-\delta)\log_q\left(\frac{1-\delta}{1-\delta+\gamma}\right)\\
&&\leq -\gamma \log_q\left(\frac{\gamma}{1-\delta+2\gamma}\right)
-(1-\delta+\gamma)\log_q\left(\frac{1-\delta+\gamma}{1-\delta+2\gamma}\right)\\
&\Leftrightarrow&
2\gamma \log_q\left(\frac{1-\delta+\gamma}{1-\delta+2\gamma}\right) \leq (1-\delta)\log_q\left(\frac{(1-\delta)(1-\delta+2\gamma)}{(1-\delta+\gamma)^2}\right)\\
&\Leftrightarrow&2(1-\delta+\gamma)\log_q (1-\delta+\gamma) \leq
(1-\delta)\log_q(1-\delta)
+(1-\delta + 2\gamma)\log_q (1-\delta + 2\gamma)
\end{eqnarray*}
The last line of the above sequence of equations is true due to the 
$f\left(\frac{a+b}{2}\right)
\leq \frac{f(a) + f(b)}{2}$ inequality that holds for the convex function $f(x) = x \log_q x$ and points $a = 1-\delta$ and $b = 1-\delta+2\gamma$.
Therefore, the claim we began with is correct.
\end{proof}
}


\newpage

\appendix
\begin{center}
\bfseries \huge Appendices
\end{center}

\section{Proof of the Convexity of $f(\gamma, \delta)$}\label{appdx:proofOfConvexity}
\pushQED{\qed}
In this section, we show that the bivariate function 
$$f(\gamma, \delta) = (1-\delta)\left[
\left(1 + \frac{\gamma}{1-\delta}\right)\log_q\frac{q(1-\delta)}{\frac{\gamma}{1-\delta}+1}
-\frac{\gamma}{1-\delta}\cdot\left(\log_{q}\frac{q(1-\delta)-1}{\frac{\gamma}{1-\delta}}\right)
\right]$$
is convex.
To prove the convexity, our general strategy is to show that the Hessian matrix of $f$ is positive semi-definite. In order to do so, we take the following steps: We first characterize a domain $D$ for $f(\gamma, \delta)$, over which we analyze the convexity. We then calculate the Hessian matrix of the function $f$, $H_f$. To show the positive semi-definiteness of $H_f$, we form its characteristic polynomial and then show that both of its solutions are real and non-negative -- meaning that both eigenvalues of $H_f$ are non-negative over the domain $D$. This would imply that $H_f$ is positive semi-definite and, hence, $f$ is convex over $D$.

\fullOnly{\paragraph{Determining the domain $D$.}}
\shortOnly{\paragraph*{Determining the domain $D$}}
Let us begin with describing the domain $D$. As stated in \Cref{sec:prior-work}, for the purposes of this paper, we only consider the error rates that are within $\delta \in [0, 1-1/q]$ and $\gamma \in [0, q-1]$. Note that for any fixed value $\delta \in [0, 1-1/q)$, $f(\gamma=0, \delta)$ is positive. We will show that as $\gamma$ grows, the value of $f(\gamma, \delta)$ continuously drops until it reaches zero at 
$\gamma = (1-\delta)(q - q\delta - 1)$. This suggests that the domain $D$ has to be defined as follows:
$$D = \left\{(\gamma, \delta)\ |\ 0\leq \delta \leq 1 - \frac{1}{q}, 0 \leq \gamma \leq (1-\delta)(q - q\delta - 1) \right\}.$$
To show the claim above, we demonstrate two simple facts: (I) The partial derivative of $f$ with respect to $\gamma$ is negative within $0 \leq \gamma \leq (1-\delta)(q - q\delta - 1)$ and, (II) $f(\gamma, \delta) = 0$ for $\gamma = (1-\delta)(q - q\delta - 1)$.

To see claim (a), note that 
$$\frac{\partial f}{\partial \gamma} = \log_q \left(\frac{(1-\delta)^2 q}{1 -\delta +\gamma}\right)-\log_q \left(\frac{(1-\delta) (q - q\delta -1)}{\gamma }\right) 
=
\log _q \left(\frac{q(1-\delta)\gamma}{(1 - \delta + \gamma)(q - q\delta - 1)}\right)
$$
which is non-positive as long as
\begin{eqnarray}
\frac{\partial f}{\partial \gamma}\leq 0 &\Leftrightarrow& \frac{q(1-\delta)\gamma}{ (1-\delta+\gamma)(q - q\delta - 1)} < 1\nonumber\\ 
&\Leftrightarrow& q(1-\delta)\gamma \leq (1-\delta+\gamma)(q - q\delta - 1) \label{eqn:derivativeSign}\\
&\Leftrightarrow&
\gamma \leq (1-\delta)(q - q\delta - 1) \nonumber
\end{eqnarray}
Note that \eqref{eqn:derivativeSign} is valid since $1-\delta+\gamma \geq 0$ and $\delta \leq 1 - \frac{1}{q} \Rightarrow q - q\delta - 1 \geq 0$.
One can also easily evaluate $f(\gamma, \delta)$ for $\gamma = (1-\delta)(q - q\delta - 1)$ to confirm claim (b).

\fullOnly{\paragraph{Hessian Matrix and Characteristic Polynomial.}}
\shortOnly{\paragraph*{Hessian Matrix and Characteristic Polynomial}}
We now proceed to calculating the Hessian matrix of $f$ and forming its characteristic polynomial.

\begin{eqnarray}
H_f&=&\begin{bmatrix}
 H_{1, 1} & H_{1, 2}\\
 H_{2, 1} & H_{2, 2}
\end{bmatrix}
=\begin{bmatrix}
 \frac{\partial^2 f}{\partial \gamma^2} & \frac{\partial^2 f}{\partial \gamma \partial \delta}\\ \\
 \frac{\partial^2 f}{\partial \delta \partial\gamma} & \frac{\partial^2 f}{\partial \delta^2}\\
\end{bmatrix}\nonumber\\
&=&\begin{bmatrix}
 \frac{1-\delta}{\gamma  (1-\delta +\gamma) \log (q)} & \frac{\gamma +(1 - \delta)^2 q}{(1-\delta) (1 - \delta +\gamma) (q - q\delta -1) \log (q)} \\ \\
 \frac{\gamma +(1-\delta)^2 q}{(1-\delta) (1 - \delta +\gamma) (q - q\delta -1) \log (q)} & \frac{(1-\delta)^3 q^2 (1 - \delta + 2 \gamma)+(2 (1-\delta) q -1)\left(\gamma ^2-\gamma  (1-\delta)-(1-\delta)^2\right)}
 {(1-\delta)^2 (1-\delta + \gamma) (q - q\delta -1)^2 \log (q)}\label{eqn:Hessian}
\end{bmatrix}
\end{eqnarray}

We prove semi-definiteness by deriving the characteristic polynomial of $H_f$. The eigenvalues of $H_f$ are the roots of this polynomial.

\begin{eqnarray}
\det(H_f - \lambda I) = 0 &\Leftrightarrow &
\left|\begin{matrix}
H_{1, 1}-\lambda & H_{1, 2}\\
 H_{2, 1} & H_{2, 2}-\lambda
\end{matrix}\right|=0\nonumber\\
&\Leftrightarrow&
(H_{1, 1}-\lambda)(H_{2, 2}-\lambda) - H_{1, 2}H_{2, 1} = 0\nonumber\\
&\Leftrightarrow&
\lambda ^ 2 - (H_{1, 1} + H_{2, 2})\lambda
 + (H_{1, 1}H_{2, 2} - H_{1, 2}H_{2, 1}) = 0\label{eqn:characteristicPoly}
\end{eqnarray}
To prove the semi-definiteness of $H_f$, we show that both of its eigenvalues are non-negative, or equivalently, the roots of the quadratic equation \eqref{eqn:characteristicPoly} are both non-negative. We remind the reader of the straightforward fact that in a quadratic equation of form $x^2 - S x + P = 0$, $S$ is the sum of the roots and $P$ is their product. Therefore, to show that both roots are non-negative, we only need to show that $S$ and $P$ are both non-negative and that the roots are both real, i.e., $\Delta = S^2-4P \geq 0$.
\begin{enumerate}
    \item $H_{1, 1} + H_{2, 2} \geq 0$
    \item $H_{1, 1}H_{2, 2} - H_{1, 2}H_{2, 1} \geq 0$
    \item $(H_{1, 1} + H_{2, 2})^2-4(H_{1, 1}H_{2, 2} - H_{1, 2}H_{2, 1}) \geq 0$
\end{enumerate}

In the remainder of this section, we prove the three items listed above.

\fullOnly{\paragraph{Proof of Item 1.}}
\shortOnly{\paragraph*{Proof of Item 1}}
Given \eqref{eqn:Hessian}, we have that 
\begin{eqnarray}
H_{1, 1} + H_{2, 2} &=& \frac{1}{\log (q)}\cdot \left(\frac{1}{\gamma }+\frac{\gamma  q^2}{(q - q\delta -1)^2}+\frac{\gamma  (1 - \gamma-\delta )}{(1-\delta)^2 (1-\delta+\gamma)}\right)\nonumber\\
&=& \frac{1}{\log (q)}\cdot \left(\frac{1}{\gamma }+\frac{\gamma }{(1 - 1/q -\delta)^2}+\frac{\gamma}{(1-\delta) (1-\delta+\gamma)}
-\frac{\gamma^2}{(1-\delta)^2 (1-\delta+\gamma)}
\right)
\end{eqnarray}
Note that terms $\frac{1}{\gamma }$ and $\frac{\gamma}{(1-\delta) (1-\delta+\gamma)}$ are positive. Therefore, to prove that $H_{1, 1} + H_{2, 2}$ is non-negative, we show that 
\begin{equation}
\frac{\gamma }{(1 - 1/q -\delta)^2}
-\frac{\gamma^2}{(1-\delta)^2 (1-\delta+\gamma)}\geq 0.\label{eqn:item1-step2}
\end{equation}
Note that 
\begin{eqnarray}
1-\delta-1/q < 1-\delta \Rightarrow 
\frac{\gamma}{(1-1/q-\delta)^2} \geq \frac{\gamma}{(1-\delta)^2}.\label{eqn:item1-step3}
\end{eqnarray}
Also, since $\delta \leq 1$, we have that $1-\delta + \gamma \geq \gamma \Rightarrow \frac{\gamma}{1-\delta + \gamma} \leq 1$. Thus, \eqref{eqn:item1-step3} holds even if one multiplies its right-hand side by $\frac{\gamma}{1-\delta + \gamma}$ which gives \eqref{eqn:item1-step2} and, thus, proves Item 1.

\fullOnly{\paragraph{Proof of Item 2.}}
\shortOnly{\paragraph*{Proof of Item 2}}
 Given \eqref{eqn:Hessian}, we have that 
$$H_{1, 1}H_{2, 2} - H_{1, 2}H_{2, 1} =
\frac{(1-\delta)^2 (q-q\delta - 1)^2
-\gamma ^2
}{\gamma  (1-\delta)^2 (1-\delta+\gamma) (q - q\delta -1)^2 \log ^2(q)}.
$$
Note that all terms in the denominator are positive. The numerator is positive as well since, as mentioned earlier, the domain $D$ is defined only to include points $(\gamma,\delta)$ where $\gamma \leq (1-\delta)(q-q\delta-1)$.

\fullOnly{\paragraph{Proof of Item 3.}}
\shortOnly{\paragraph*{Proof of Item 3}}
This claim can be simply shown to be true as follows: $$ (H_{1, 1} + H_{2, 2})^2-4(H_{1, 1}H_{2, 2} - H_{1, 2}H_{2, 1}) = 
(H_{1, 1} - H_{2, 2})^2 + 4H_{1, 2}H_{2, 1} = 
(H_{1, 1} - H_{2, 2})^2 + 4H^2_{1, 2}
$$
The final term is trivially positive. Note that the last step follows from the fact that $H_{1, 2}=H_{2, 1}$.
\popQED

\shortOnly{
\section{Missing Figures}
We have moved two large pictures to this section in the appendix to adhere with the 15-page conference submission limit. Please find \cref{fig:q=2,fig:q=5-intro} in the following pages.
\FancyPicture
\FigureNiceDeltas
}

\newpage

\bibliographystyle{plain}
\bibliography{bibliography}

\begin{thebibliography}{10}

\bibitem{blawat2016forward}
Meinolf Blawat, Klaus Gaedke, Ingo Huetter, Xiao-Ming Chen, Brian Turczyk,
  Samuel Inverso, Benjamin~W Pruitt, and George~M Church.
\newblock Forward error correction for {DNA} data storage.
\newblock {\em Procedia Computer Science}, 80:1011--1022, 2016.

\bibitem{bornholt2016dna}
James Bornholt, Randolph Lopez, Douglas~M Carmean, Luis Ceze, Georg Seelig, and
  Karin Strauss.
\newblock A {DNA}-based archival storage system.
\newblock {\em ACM SIGARCH Computer Architecture News}, 44(2):637--649, 2016.

\bibitem{brakensiek2016efficient}
Joshua Brakensiek, Venkatesan Guruswami, and Samuel Zbarsky.
\newblock Efficient low-redundancy codes for correcting multiple deletions.
\newblock {\em IEEE Transactions on Information Theory}, 64(5):3403--3410,
  2018.

\bibitem{bukh2017improved}
Boris Bukh, Venkatesan Guruswami, and Johan H{\aa}stad.
\newblock An improved bound on the fraction of correctable deletions.
\newblock {\em IEEE Transactions on Information Theory}, 63(1):93--103, 2017.

\bibitem{cheng2020efficient}
Kuan Cheng, Venkatesan Guruswami, Bernhard Haeupler, and Xin Li.
\newblock Efficient linear and affine codes for correcting
  insertions/deletions.
\newblock In {\em Proceedings of the 2021 ACM-SIAM Symposium on Discrete
  Algorithms (SODA)}, pages 1--20, 2021.

\bibitem{cheng2019synchronization}
Kuan Cheng, Bernhard Haeupler, Xin Li, Amirbehshad Shahrasbi, and Ke~Wu.
\newblock Synchronization strings: highly efficient deterministic constructions
  over small alphabets.
\newblock In {\em Proceedings of the 30th Annual ACM-SIAM Symposium on Discrete
  Algorithms (SODA)}, pages 2185--2204, 2019.

\bibitem{cheng2018deterministic}
Kuan Cheng, Zhengzhong Jin, Xin Li, and Ke~Wu.
\newblock Deterministic document exchange protocols, and almost optimal binary
  codes for edit errors.
\newblock In {\em Proceedings of the 59th Annual {IEEE} Symposium on
  Foundations of Computer Science (FOCS)}, pages 200--211, 2018.

\bibitem{cheng2019block}
Kuan Cheng, Zhengzhong Jin, Xin Li, and Ke~Wu.
\newblock Block edit errors with transpositions: Deterministic document
  exchange protocols and almost optimal binary codes.
\newblock In {\em Proceedings of the 46th International Colloquium on Automata,
  Languages, and Programming (ICALP)}, volume 132 of {\em LIPIcs}, pages
  37:1--37:15. Schloss Dagstuhl - Leibniz-Zentrum f{\"{u}}r Informatik, 2019.

\bibitem{cheraghchi2019overview}
Mahdi Cheraghchi and Jo{\~a}o Ribeiro.
\newblock An overview of capacity results for synchronization channels.
\newblock {\em IEEE Transactions on Information Theory}, 2020.

\bibitem{church2012next}
George~M Church, Yuan Gao, and Sriram Kosuri.
\newblock Next-generation digital information storage in {DNA}.
\newblock {\em Science}, 337(6102):1628--1628, 2012.

\bibitem{liu2019explicit}
Tai Do~Duc, Shu Liu, Ivan Tjuawinata, and Chaoping Xing.
\newblock Explicit constructions of two-dimensional reed-solomon codes in high
  insertion and deletion noise regime.
\newblock {\em IEEE Transactions on Information Theory}, 67(5):2808--2820,
  2021.

\bibitem{goldman2013towards}
Nick Goldman, Paul Bertone, Siyuan Chen, Christophe Dessimoz, Emily~M LeProust,
  Botond Sipos, and Ewan Birney.
\newblock Towards practical, high-capacity, low-maintenance information storage
  in synthesized {DNA}.
\newblock {\em Nature}, 494(7435):77, 2013.

\bibitem{guruswami2004list}
Venkatesan Guruswami.
\newblock {\em List decoding of error-correcting codes: winning thesis of the
  2002 ACM doctoral dissertation competition}, volume 3282.
\newblock Springer Science \& Business Media, 2004.

\bibitem{guruswami2019optimally}
Venkatesan Guruswami, Bernhard Haeupler, and Amirbehshad Shahrasbi.
\newblock Optimally resilient codes for list-decoding from insertions and
  deletions.
\newblock In {\em Proceedings of the 52nd Annual {ACM} {SIGACT} Symposium on
  Theory of Computing (STOC)}, pages 524--537, 2020.

\bibitem{guruswami2022zero}
Venkatesan Guruswami, Xiaoyu He, and Ray Li.
\newblock The zero-rate threshold for adversarial bit-deletions is less than
  1/2.
\newblock In {\em Proceedings of the IEEE Symposium on Foundations of Computer
  Science (FOCS)}, pages 727--738, 2022.

\bibitem{guruswami2016efficiently}
Venkatesan Guruswami and Ray Li.
\newblock Efficiently decodable insertion/deletion codes for high-noise and
  high-rate regimes.
\newblock In {\em Proceedings of the {IEEE} International Symposium on
  Information Theory {(ISIT)}}, pages 620--624, 2016.

\bibitem{guruswami2017deletion}
Venkatesan Guruswami and Carol Wang.
\newblock Deletion codes in the high-noise and high-rate regimes.
\newblock {\em IEEE Transactions on Information Theory}, 63(4):1961--1970,
  2017.

\bibitem{haeupler2019optimal}
Bernhard Haeupler.
\newblock Optimal document exchange and new codes for insertions and deletions.
\newblock In {\em Proceedings of the 60th Annual {IEEE} Symposium on
  Foundations of Computer Science (FOCS)}, pages 334--347, 2019.

\bibitem{haeupler2019near}
Bernhard Haeupler, Aviad Rubinstein, and Amirbehshad Shahrasbi.
\newblock Near-linear time insertion-deletion codes and
  (1+$\varepsilon$)-approximating edit distance via indexing.
\newblock In {\em Proceedings of the 51st Annual {ACM} {SIGACT} Symposium on
  Theory of Computing (STOC)}, pages 697--708, 2019.

\bibitem{haeupler2017synchronization3}
Bernhard Haeupler and Amirbehshad Shahrasbi.
\newblock Synchronization strings: Explicit constructions, local decoding, and
  applications.
\newblock In {\em Proceedings of the 50th Annual {ACM} {SIGACT} Symposium on
  Theory of Computing (STOC)}, pages 841--854, 2018.

\bibitem{haeupler2020survey}
Bernhard Haeupler and Amirbehshad Shahrasbi.
\newblock Synchronization strings and codes for insertions and deletions--a
  survey.
\newblock {\em IEEE Transactions on Information Theory}, 2021.

\bibitem{haeupler2017synchronization}
Bernhard Haeupler and Amirbehshad Shahrasbi.
\newblock Synchronization strings: Codes for insertions and deletions
  approaching the {S}ingleton bound.
\newblock {\em Journal of the ACM (JACM)}, 68(5):1--39, 2021.

\bibitem{haeupler2018synchronization4}
Bernhard Haeupler, Amirbehshad Shahrasbi, and Madhu Sudan.
\newblock Synchronization strings: List decoding for insertions and deletions.
\newblock In {\em Proceedings of the 45th International Colloquium on Automata,
  Languages, and Programming (ICALP)}, volume 107 of {\em LIPIcs}, pages
  76:1--76:14. Schloss Dagstuhl - Leibniz-Zentrum f{\"{u}}r Informatik, 2018.

\bibitem{haeupler2017synchronization2}
Bernhard Haeupler, Amirbehshad Shahrasbi, and Ellen Vitercik.
\newblock Synchronization strings: Channel simulations and interactive coding
  for insertions and deletions.
\newblock In {\em Proceedings of the 45th International Colloquium on Automata,
  Languages, and Programming (ICALP)}, volume 107 of {\em LIPIcs}, pages
  75:1--75:14. Schloss Dagstuhl - Leibniz-Zentrum f{\"{u}}r Informatik, 2018.

\bibitem{hayashi2018list}
Tomohiro Hayashi and Kenji Yasunaga.
\newblock On the list decodability of insertions and deletions.
\newblock In {\em Proceedings of the {IEEE} International Symposium on
  Information Theory {(ISIT)}}, pages 86--90, 2018.

\bibitem{levenshtein1974elements}
Vladimir~I Levenshtein.
\newblock Elements of coding theory.
\newblock {\em Diskretnaya matematika i matematicheskie voprosy kibernetiki},
  pages 207--305, 1974.

\bibitem{liu2019listarxiv}
Shu Liu, Ivan Tjuawinata, and Chaoping Xing.
\newblock On list decoding of insertion and deletion errors.
\newblock {\em CoRR}, abs/1906.09705, 2019.

\bibitem{liu2019list}
Shu Liu, Ivan Tjuawinata, and Chaoping Xing.
\newblock Efficiently list-decodable insertion and deletion codes via
  concatenation.
\newblock {\em IEEE Transactions on Information Theory}, 67(9):5778--5790,
  2021.

\bibitem{mercier2010survey}
Hugues Mercier, Vijay~K Bhargava, and Vahid Tarokh.
\newblock A survey of error-correcting codes for channels with symbol
  synchronization errors.
\newblock {\em IEEE Communications Surveys \& Tutorials}, 12(1):87--96, 2010.

\bibitem{mitzenmacher2009survey}
Michael Mitzenmacher.
\newblock A survey of results for deletion channels and related synchronization
  channels.
\newblock {\em Probability Surveys}, 6:1--33, 2009.

\bibitem{organick2017scaling}
Lee Organick, Siena~Dumas Ang, Yuan-Jyue Chen, Randolph Lopez, Sergey Yekhanin,
  Konstantin Makarychev, Miklos~Z Racz, Govinda Kamath, Parikshit Gopalan,
  Bichlien Nguyen, et~al.
\newblock Scaling up {DNA} data storage and random access retrieval.
\newblock {\em BioRxiv}, page 114553, 2017.

\bibitem{wachter2017list}
Antonia Wachter-Zeh.
\newblock List decoding of insertions and deletions.
\newblock {\em IEEE Transactions on Information Theory}, 64(9):6297--6304,
  2018.

\bibitem{yazdi2015dna}
SM~Hossein~Tabatabaei Yazdi, Han~Mao Kiah, Eva Garcia-Ruiz, Jian Ma, Huimin
  Zhao, and Olgica Milenkovic.
\newblock {DNA}-based storage: Trends and methods.
\newblock {\em IEEE Transactions on Molecular, Biological and Multi-Scale
  Communications}, 1(3):230--248, 2015.

\end{thebibliography}

\end{document}